\documentclass[pdflatex,sn-mathphys-num]{sn-jnl}% Math and Physical Sciences Numbered Reference Style
\usepackage{stmaryrd} %%% For double square brackets %%%

\newcommand{\bZ}{\mathbb{Z}}

\newcommand{\boldx}{\mathbf{x}}

\newcommand{\boldX}{\mathbf{X}}
\newcommand{\boldY}{\mathbf{Y}}

\usepackage{tikz}
\usepackage{graphicx}%
\usepackage{multirow}%
\usepackage{amsmath,amssymb,amsfonts}%
\usepackage{amsthm}%
\usepackage{mathrsfs}%
\usepackage[title]{appendix}%
\usepackage{xcolor}%
\usepackage{textcomp}%
\usepackage{manyfoot}%
\usepackage{booktabs}%
\usepackage{algorithm}%
\usepackage{algorithmic}%
\usepackage{listings}%
\newcommand{\range}[1]{\llbracket #1 \rrbracket}
\newcommand{\zrange}[1]{\llbracket #1 \rrbracket}
\newcommand{\red}[1]{\textcolor{red}{#1}}

\makeatletter
\def\namedlabel#1#2{\begingroup
    #2%
    \def\@currentlabel{#2}%
    \phantomsection\label{#1}\endgroup
}
\makeatother

\usepackage{verbatim}
\usepackage{amsthm,amssymb}
\usepackage{xcolor}
\usepackage{graphicx}
\usepackage{amssymb}
\usepackage{epstopdf}
\usepackage{hyperref}
\hypersetup{
 colorlinks,
 linkcolor={blue!100!black},
 citecolor={blue!100!black},
 urlcolor={blue!80!black}
}
\newif\ifFULL
\FULLtrue

\usepackage[utf8]{inputenc} 
\usepackage[T1]{fontenc}
\usepackage{url}

\usepackage{ifthen}
\usepackage[cmex10]{mathtools}
\newtheorem{theorem}{Theorem}

\newtheorem{remark}{Remark}

\newtheorem{proposition}{Proposition}

\newtheorem{lemma}{Lemma}
\newtheorem{definition}{Definition}

\newtheorem{corollary}{Corollary}

\begin{document}

\title[Single Fragment Forensic Coding\\from Discrepancy Theory]{Single Fragment Forensic Coding from Discrepancy Theory}


\author{\fnm{Junsheng} \sur{Liu}}\email{junsheng@wustl.edu}
\author{\fnm{Netanel} \sur{Raviv}}\email{netanel.raviv@wustl.edu}

\affil{\orgdiv{Department of Computer Science and Engineering}, \orgname{Washington University in St.\ Louis}, \city{St.\ Louis}, \state{MO}, \country{USA}}


\abstract{Three-dimensional (3D) printing's accessibility enables rapid manufacturing but also poses security risks, such as the unauthorized production of untraceable firearms and prohibited items. 
    To ensure traceability and accountability, embedding unique identifiers within printed objects is essential, in order to assist forensic investigation of illicit use.
    This paper models data embedding in 3D printing using principles from error-correcting codes, aiming to recover embedded information from partial or altered fragments of the object. 
    Previous works embedded one-dimensional data (i.e., a vector) inside the object, and required almost all fragments of the object for successful decoding.
    In this work, we study a problem setting in which only one sufficiently large fragment of the object is available for decoding. 
    We first show that for one-dimensional embedded information the problem can be easily solved using existing tools.
    Then, we introduce novel encoding schemes for two-dimensional information (i.e., a matrix), and three-dimensional information (i.e., a cube) which enable the information to be decoded from any sufficiently large rectangle-shaped or cuboid-shaped fragment.
    Lastly, we introduce a code that is also capable of correcting bit-flip errors, using techniques from recently proposed codes for DNA storage.
    %.that occur on any sufficiently large fragment. 
    Our codes operate at non-vanishing rates,
    and involve concepts from discrepancy theory called Van der Corput sets and Halton-Hammersely sets in novel ways.}




\maketitle

\section{Introduction}
%\footnotemark
%\footnotetext{
{\let\thefootnote\relax\footnote{Parts of this work were presented in the 2025 IEEE International Symposium on Information Theory. This work was supported in part by NSF grant CNS-2223032.}}
3D printing technology has revolutionized manufacturing by enabling on-demand production using commodity printers. 
However, its accessibility and versatility also present significant security challenges, particularly concerning the unauthorized or untraceable production of firearms and other prohibited items. 
Addressing these challenges requires innovative mitigation strategies that ensure traceability and accountability without hindering the legitimate uses of 3D printing. 

To counteract these security challenges, various information embedding techniques are employed to tag objects with identifying information, so they can be traced and verified during forensic investigation. 
For instance, physical modifications to the printing process (e.g., altering layer width or orientation) can affect the object in a way that embeds bits into the object without changing its material properties.
While multiple effective embedding techniques have been proposed in the literature (e.g.~\cite{delmotte2019blind,narendra2022watermarking}), including secure mechanisms for embedding information using untrusted printers~\cite{wang2024secureinformationembeddingextraction}, the embedded information is exposed to tampering by the malicious actor which owns both the printer and the object.
%While embedding techniques are effective in enhancing the traceability and security of 3D printed objects, they face inherent challenges that parallel issues encountered in error-correcting codes within coding theory. 
%Understanding these challenges provides insight into improving the robustness of security measures.

For example, malicious actors can attempt to erase, modify, or obscure the embedded bits to evade detection and traceability. 
While some actions taken by the malicious actor create errors of types that are well-studied within coding theory, this problem setting gives rise to unique threats that were not studied before.
Our paper focuses on one such threat, in which the malicious actor breaks the object apart and potentially conceals some of the resulting fragments.

We model this problem as one of communication under noise, where an encoder (i.e., the printer) transmits a codeword (i.e., a unique identifier) through an adversarial channel, in which the adversary (i.e., the malicious actor) may break it to pieces and conceal some of them, while being restricted by some security parameters.
The decoder (i.e., law-enforcement) receives only one connected fragment of this codeword. 
In addition, the encoding and decoding process described above might induce bit-flips due to improper reading or writing of bits into the object, hence bit-flip resilience is necessary.
%In the above writing and reading process, bit-flips errors can occur as tolerance is required for 3D printer to print embedded information.
A code is considered effective if the decoder can accurately reconstruct the original message from the received fragment with bit-flip errors, under any adversarial behavior which falls within the security parameters.
This \textit{forensic-coding} model, which can be seen as a variant of the recently proposed torn-paper model (\cite{barlev,wang2024breakresilientcodesforensic3d,wang2024secureinformationembeddingextraction,tornpaper}, see next) provides a realistic abstraction for the information embedding problem in 3D printing. 

The forensic coding model can be instantiated using various bit embedding schemes, which can embed either a vector, a matrix, or even a three-dimensional cube, each of which resulting in different system requirements and information density. 
The vector case can be addressed using the well-studied concept of
the cyclically permutable codes~\cite{QA1992ConstructionsOB}. 
%In the vector case, t
These codes ensure that any connected fragment of sufficient length contains enough information to reconstruct the original message (details in Section~\ref{section:problem definition for 2D}). 

However, when extending this problem to matrices or cubes, new challenges emerge due to the increased complexity of possible fragment shapes and sizes. 
In both cases, fragments can vary not only in length but also in width and overall shape, leading to a vast number of potential fragment configurations. 
%Our paper separates the problem into different cases.
For ease of presentation, our paper separates the problem to its two-dimensional (i.e., a matrix data) and three-dimensional (i.e., cube data) variants, which are treated separately in Section~\ref{Section:pre for 2D} and in Section~\ref{section: pre for 3D}, respectively. 
%In Section~\ref{Section:pre for 2D},~\ref{section:encoding for 2D} and~\ref{section:decoding for 2D}, we develop matrix encoding and decoding, which accommodates decoding of fragments which contain a sufficiently large rectangle.
In Section~\ref{Section:pre for 2D} we develop matrix encoding which accommodates decoding from any fragment which contains a sufficiently large axis-parallel rectangle with some minimum side length.
In Section~\ref{section: pre for 3D} we develop cube encoding which accommodates decoding from any fragment which contain sufficiently large cuboid, again, with minimum side length.
Further, in Section~\ref{section:error correction} we develop matrix encoding and decoding which can additionally tolerate bit-flips; this is done using techniques from the recently introduced sliced-channel~\cite{sima2024robust}.
%for fragments containing bit-flips errors, as an extension for error-free case. 

Our results make use of concepts from discrepancy theory~\cite{chen2014panorama}, a mathematical theory which studies how evenly elements can be distributed across combinatorial or geometric arrangements.
We use known constructions of such arrangements in order to disperse information evenly throughout the matrix (or cube), so that any fragment contains all decodable information.
In particular, Section~\ref{Section:pre for 2D} uses a concept called Van der Corput sets~\cite{zbMATH02531767}, a set of points in a two dimensional square which minimizes the area of the largest empty rectangle.
Section~\ref{section: pre for 3D} uses an extension of this concept to three dimensions due to Halton~\cite{halton1960efficiency}.

%Our paper addresses this by exploring the properties of effective codes under different fragment shapes, such as various rectangles with the same area. 
%\red{[Paragraph to guide the reader about 1D, and then 2D, and then the problem of shapes, and how we address it.]}

%By exploring the properties of effective codes under different fragment shapes, such as various rectangles with the same area, we can develop robust methods for embedding and recovering data in 3D printed objects. 

Previous works about forensic coding addressed a similar adversary which can break the object (and hence the information) apart, but focused on the less realistic case where (almost) all fragments are given to the decoder, and the less information-dense case of embedding vectors~\cite{wang2024breakresilientcodesforensic3d} (see also~\cite{wang2024secureinformationembeddingextraction}). 
Further, similar problems have been studied in the area of DNA storage~\cite{sima2023error}. 
For example,~\cite{tornpaper} introduced \textit{torn-paper coding}, designed to enhance data reliability in channels where information fragments may be torn and reordered by a \textit{probabilistic} (i.e., not adversarial) process.
Closer to the current paper,~\cite{barlev} presents \textit{adversarial} torn-paper codes, extending torn-paper coding to worst-case scenarios. However,~\cite{barlev} focused on decoding vector information from almost all fragments, and restricted the fragments in length.
%to achieve optimal rates and efficient encoding and decoding for DNA data storage. 
Additionally,~\cite{sima2023robustindexingslicedchannel,sima2024robust} studies a channel model in which information is sliced at prescribed locations, which is less relevant to our target application. 
Our paper extends the above line of works and provides a new scheme for single-fragment decoding of matrix or cube data (Fig.~\ref{figure:problemDefinition}).
% Our goal is to recover information coded as a matrix or cube that is broken into fragments, and at least one fragment which contains a rectangle or cuboid whose area or volume is at least~$M$ (Fig.~\ref{figure:problemDefinition}) is available to the decoder. 
% We provide a detailed analysis for both case, with encoding methods reaching the best possible rate asymptotically.

\begin{figure}[h]
    \centering    
 \includegraphics[width=0.9\linewidth]{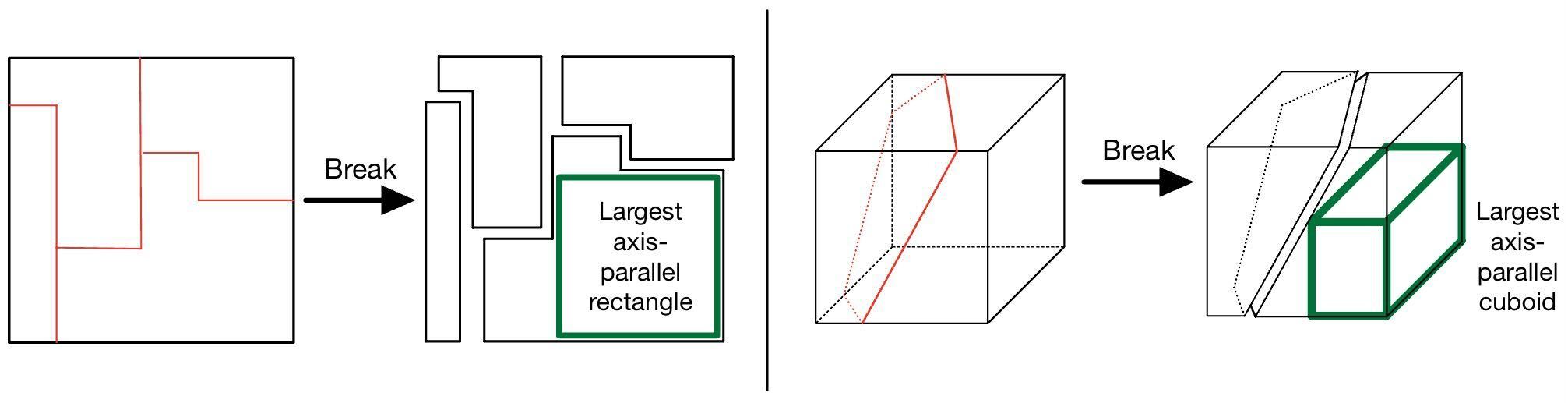}

\caption{\textbf{(Left)} Two dimensional embedded information (a matrix) is broken along the red lines, creating a collection of fragments, only one of which is received by the decoder. 
Our schemes guarantee correct decoding as long as the recieved fragment contains an axis-parallel rectangle of a certain minimum area and a certain minimum side length.
This problem setting is discussed in Section~\ref{Section:pre for 2D}.
\textbf{(Right)} A similar problem setting with three dimensional embedded information (a cuboid), see Section~\ref{section: pre for 3D}.}\label{figure:problemDefinition}
\end{figure}
%\tikzset{every picture/.style={line width=0.75pt}} %set default line width to 0.75pt  

%\section{Preliminaries for 2D}\label{Section:pre for 2D}
\section{Forensic coding for two dimensional information}\label{Section:pre for 2D}

\subsection{Problem Definition}\label{section:problem definition for 2D}
In order to embed a fingerprint in an object, we formulate the problem as one of communication under adversarial noise. 
To this end, we model the printer as an encoding mechanism (or encoder), which maps a message~$\boldx$ containing the information to be embedded (e.g., printer ID, geolocation, etc.), into a codeword. 
Since the vector case is easily addressable using cyclically permutable codes\footnote{A cyclically permutable code (CPC,~\cite{QA1992ConstructionsOB}) is one where no codeword is a cyclic permutation of another, and they can be constructed~\cite{chen2023enumeration} from cyclic codes~\cite[Ch.~8]{roth2006introduction}.
In vector forensic codes, where one fragment of length at least~$M$ is available to the decoder, it is an easy exercise to verify that an optimal construction is given by encoding the information using a CPC of length~$M$ and then repeating it~$n/M$ times to form a codeword.
To the best of our knowledge, CPCs do not directly extend to our
two dimensional setting.
%The root cause is that the ring $\mathbb{F}_{q}[x]/(x^{n}-1)$, over which vector cyclic (and hence cyclically permutable~\cite{chen2023enumeration}) codes are constructed, is much more well-behaved than its two dimensional counterpart~$\mathbb{F}_{q}[x,y]/(x^{n}-1,y^{n}-1)$; the former is a principal ideal domain, while the latter is not.
%In one dimension the only shift to consider is a cyclic rotation of the $n$ coordinates, and the ambient quotient ring $\mathbb{F}_{q}[x](x^{n}-1)$ is a principal ideal ring; hence every linear cyclic code is generated by a single polynomial and the orbit of a codeword under cyclic shifts has cardinality~$n$.
%For two dimensions the natural analogue is
%\[
%   \mathbb{F}_{q}[x,y]/(x^{n}-1,y^{n}-1),
% \]
%which is \emph{not} a principal ideal ring.
%\purple{Consequently, the combined horizontal and vertical shift action fails to yield an orbit of size
%$n^{2}$, so CPCs cannot be improved in higher dimnension setting}   
Applying one-dimensional CPC row by row is possible but highly inefficient as it results in vanishing rates.}, we focus on matrix codewords.
Specifically, a codeword in this setting is a matrix~$C\in\Sigma_q^{n\times n}$, where~$n$ is a parameter of the system and~$\Sigma_q$ is an alphabet of size $q$.
We consider the message~$\boldx$ as residing in~$\Sigma_q^k$ for some parameter~$k$.

%We model the perpetrator as an adversary which receives~$C$, breaks it apart, and discards of all fragments except one. 
After encoding, an adversary receives~$C$, breaks it apart, and discards of all fragments except one. 
That one fragment arrives at the decoder (i.e., law enforcement) without any noise (the noisy setting is studied separately in Section~\ref{section:error correction}), and the decoder must reconstruct the message~$\boldx$ in all cases. 
Clearly, this task is impossible without clear restrictions on the power of the adversary, e.g., the size and shape of the received fragment. 
To this end, we impose the restriction that the fragment contains a rectangle of area at least~$M$, where~$M$ is an additional parameter of the system. 
The discussion is not limited if the decoder receives multiple fragments, or a fragment which contains larger rectangles, as long as there exists one fragment which contains a rectangle of area at least~$M$. 

Due to the challenging nature of the problem, we further assume that the received fragment contains a rectangle of area~$M$ which is neither ``too thin'' nor ``too thick,'' i.e., the received fragment contains an~$a\times b$ rectangle for some integers~$a$ and~$b$ such that~$ab\ge M$ and~$\min\{a,b\}\ge h$ for an additional parameter~$h$; such a fragment is called \textit{legal}. 

Recently, Sun and Lu~\cite{sun2025criss} studied a setting in which an~$n\times n$ array code suffers a \emph{$(t_r,t_c)$-criss–cross deletion}---the simultaneous removal of~$t_r$ consecutive rows and~$t_c$ consecutive columns.  
By employing non-binary Varshamov–Tenengolts (VT) codes~\cite{varshamov1965codes}, they attain nearly optimal redundancy, albeit under specific constraints on $t_r$, $t_c$, and on the alphabet size~$q$.
In this paper we tackle a closely related yet distinct scenario.  
Our problem can be considered as constructing a code which simultaneously accommodates any parameters~$t_r,t_c$ such that
\begin{align*}
  (n-t_r)(n-t_c) &\;\ge\; M, \\%\label{eq:capacity} \\
  \min\{\,n-t_r,\;n-t_c\,\} &\;\ge\; h, %\label{eq:height}
\end{align*} 
and additionally, assumes that row deletions occur in two bursts, one at the top and one at the bottom of the matrix (and similarly for columns).
%where~$t_r$ (resp.~$t_c$) denotes the number of deleted rows (resp. columns) and $h$ is a prescribed height parameter.
%These conditions differ from~\cite{sun2025criss} in two key directions: (a) the deleted rows and columns may surround the fragment rather \blue{be consecutive}; and (b) 
Further, we allow $t_r,t_c=O(n)$ with \emph{no} restriction on~$q$.
Hence, our constructions solve a different problem at a wider parameter regime.
%This additional condition is not overly restrictive since thin or thick fragments are harder to obtain by the adversary from a mechanical strength standpoint. 

Our goal is constructing encoding and decoding functions which enable the decoder to retrieve the message~$\boldx$ in all cases where the received fragment contains a rectangle of area at least~$M$ with sides of length at least~$h$. 
For brevity, we call such codes~\textit{$(q,n,M,h)$-forensic}, or \textit{forensic codes} for short. 
The figure of merit by which we assess a given forensic code is the \textit{rate}~$k/M$. 
Notice that we normalize by~$M$ (fragment area) rather than by~$n^2$ (codeword area) since the decoder only receives~$M$ bits to decode.

\subsection{Overview}\label{section:overview}
%\red{[Find a place in the encoding section to put this sentence:]} \blue{For convenience we also assume that~$m|n$ and~$d|n$, although the construction can be generalized to other cases as well.}
%We construct forensic codes using ideas from discrepancy theory, a mathematical theory which studies inevitable irregularities which must emerge in different types of distributions.

%The challenge in constructing forensic codes is twofold: 
There are two challenges in constructing forensic codes:
\begin{enumerate}
    \item[\namedlabel{item:enoughInfo}{(I)}] One must guarantee that every legal fragment contains sufficient information in order to reconstruct the original message. 
    \item[\namedlabel{item:alignment}{(II)}] The positioning (or alignment) of the received fragment with respect to the codeword is unknown.
\end{enumerate}
To address these challenges, we begin by partitioning the~$n\times n$ codeword into smaller~$d\times d$ submatrices called \textit{units}. 

To address~\ref{item:enoughInfo} we use a known discrepancy-theoretic notion called Van der Corput sets~\cite{Dumitrescu_Jiang_2012} in order to disperse the information in a uniform manner, so that every sufficiently large rectangle of units contains the information in its entirety. 
In more detail, we address~\ref{item:enoughInfo} by associating the different units in a codewords with different \textit{colors}, i.e., each unit is mapped to an integer, and units of the same color contain the same information. 
The colors are chosen in accordance with the properties of Van der Corput sets so that any fragment of at least a certain area will contain all different colors, and hence all different types of units, and the message could be decoded.

To address challenge~\ref{item:alignment} we allocate one designated unit to contain a certain fixed pattern (rather than information), that is discernible from the content of any other unit.
That distinguished unit serves as a marker by which the position of all other units is determined. 
The presence of such unit in every legal fragment is again guaranteed by the properties of Van der Corput sets.

% For simplicity, please check detailed proof in decoding part. 

%To guarantee that every large enough fragment contains all colors while maximizing information content, we make use of Van der Corput sets from discrepancy theory. 
Van der Corput sets are finite sets of points in the two-dimensional unit square such that every axis-parallel rectangle of at least a given area must contain a point. 
Largely speaking, we use variants of Van der Corput sets as color sets for units, and hence large enough fragments will correspond to large enough axis parallel rectangles, which will therefore contain all colors. 
Section~\ref{section:VDC} below introduces Van der Corput sets and their variants which are necessary for constructing forensic codes.
Section~\ref{section:encoding for 2D} and Section~\ref{section:decoding for 2D} which follow introduce the resulting encoding and decoding mechanisms.

\subsection{Van der Corput sets}\label{section:VDC}
%As mentioned in the previous section, the Van der Corput set plays a key roll in coloring the units. 
%This section presents definitions and lemma pertaining to the Van der Corput set, which are used in the encoding process.

To introduce the Van der Corput set, we begin by defining the bit reversing function. 
\begin{definition}[Bit reversing function]\label{definition:Bit reversing function}
    Let $z$ be  an integer $0\leq z<2^c$ for some integer $c\in \bZ^+$, and let~$(a_0,\ldots,a_{c-1})$ be the binary representation of~$z$, i.e., $z=\sum_{i=0}^{c-1}a_i2^i$. Then, the bit reversing function~$f_{2^c}$ is defined as $f_{2^c}(z)=\sum_{i=0}^{c-1}a_i2^{c-1-i}$. 
\end{definition}
The bit reversing function can be viewed as reversing the bits of the binary representation for an integer given as input. 
For example, for~$c=5$ the binary representation of~$z=13$ is $01101$, and hence $f_{32}(z)=22$, since the reverse of $01101$ is $10110$. 
Notice that the bit reversing function is bijective.
Using the bit reversing function we define the Van der Corput set. 
\begin{definition}[Van der Corput set~\cite{zbMATH02531767,Dumitrescu_Jiang_2012}]\label{definition:VDCset}
    For a positive integer power of two~$w$, the $w$ Van der Corput set ($w$-VDC set, for short) is~$C_w\triangleq\{(q,f_w(q))\}_{q=0}^{w-1}$. 
\end{definition}
The Van der Corput set is a crucial notion in discrepancy theory, especially for investigating the largest empty axis-parallel rectangle for a given set of points. 
Specifically, for a finite set of points~$T$ in $R=[0,r]\times[0,r]$, with $r$ being any positive integer, 
let~$A_r(T)$ be the area of the largest empty integer axis-parallel rectangle\footnote{An integer axis-parallel rectangle is one in which the vertices have integer values, and all edges are parallel to one of the axes.} (abbrv. largest empty rectangle) contained in~$R$. 

The quantity~$A_r(T)$ is a fundamental property of distributions of sets in space, as evenly spaced points will result in a smaller~$A_r(T)$. 
The challenge is to bound~$A_r(T)$ for any set of~$w$ points~\cite{Dumitrescu_Jiang_2012}, and to construct sets with small~$A_r(T)$ explicitly. 
One such explicit set is the Van der Corput set (Definition~\ref{definition:VDCset}). 
The following lemma is a special case of a well known one; 
\ifFULL
the proof is given in full due to minor differences and extensions that will be discussed later.
\else
a full proof is given in~\cite{combinedWebsites} due to space constraints.
\fi
In what follows we denote $\range{n}\triangleq\{0,1,\ldots,n\}$.
\begin{lemma}\label{lemma:VDCarea}
    For a positive integer~$c$ and~$w=2^c$, the Van der Corput set $C_w\subseteq[0,w]\times[0,w]$ satisfies $A_w(C_w)<4w$.
\end{lemma}
\ifFULL
    \begin{proof}
    Let~$B\subseteq[0,w]\times[0,w]$ be an empty rectangle. 
    To prove the theorem we ought to show that the area of~$B$ is smaller than~$4w$. Since~$B$ is a rectangle, there exists two segments~$S_x,S_y\subseteq[0,w]$ such that~$B=S_x\times S_y$.
    
    For a positive integer~$v\in\range{c}$ and positive integer $u\in \range{2^{v}-1}$, 
    %\red{[$u$ here is continuous? Meaning, ``canonical''$=$``of length which is a power of two''? The proof does not appear to be true if~$u$ is continuous, since in order for two numbers to agree on (say) two MSBs they need to both be in either~$[0,0.25),[0.25,0.5),[0.5,0.75),[0.75,1)$.} 
    the interval $[u\cdot 2^{c-v}, (u+1)\cdot 2^{c-v})$ is called \textit{canonical}. 
    Let~$I$ be the longest canonical interval in~$S_y$, and denote $I=[a\cdot 2^{c-b}, (a+1)\cdot 2^{c-b})$ for some integers~$a$ and~$b$.
    Notice that since~$I$ is the longest canonical interval contained in~$S_y$, it follows that the length of~$S_y$ is strictly less than $2\cdot 2^{c-b+1}$; 
    %\blue{Here the longest canonical inteval will be length 0.25 and the length of canonical interval will be $2^{c-b}$. We are taking $2\times 2^{c-b+1}$ which will be 1 in this case, so  $S_y<1$ which is fine} 
    %\red{[Not~$3\cdot 2^{c-b+1}$? What about~$S_y=[0.25+\epsilon,1-\epsilon]$, $\epsilon<1/8$? The longest canonical interval seems to be~$[0.5,0.75)$, but the length is larger than~$1/2$.]}; 
    otherwise, $S_y$ would have contained $[a'\cdot 2^{c-b+1},(a'+1)2^{c-b+1}]$ for some~$a'\in[0,2^v-1]$. 
    %\red{[It seems that~$b,c$ are the same as above, but~$a$ could be different, i.e., it should have been~``$S_y$ would have contained~$[a'\cdot 2^{c-b+1},(a'+1)2^{c-b+1}]$ for some~$a'\in[0,2^v-1]$.''}

    Let~$D\subseteq C_w$ be the set of points whose~$y$ coordinates are in~$I$.
    First, by the definition of~$C_w$, it follows that the~$b$ most significant bits of~$y$ are identical to the~$b$ least significant bits of~$x$, for any~$(x,y)\in C_w$.
    Second, observe that all~$y$'s such that~$(x,y)\in D$ agree on their~$b$ most significant bits.
    %\red{[This seems to be true only if both~$u$ and~$v$ are integers.]}.
    Therefore, since the bit reversing function is bijective, the distance between the~$x$ coordinates of every two adjacent points in~$D$, when sorted by the $x$ coordinate, is~$2^b$.    
    Since~$B$ is an empty rectangle, it follows that the length of~$S_x$ is less than~$2^b$; otherwise $B$ will contain a point.
    Therefore, $A_w(C_w)<2^b\cdot 2\cdot 2^{c-b+1}=4w$. \qedhere
\end{proof}
Therefore, the van der Corput set is an explicit set whose largest empty rectangle is small. 
In the sequel, we will require the following generalization of VDC sets.
\fi
 
\begin{definition}[Shifted Van der Corput set]\label{Definition:shifted vdc set}
    
    For a VDC set $C_w=\{(q,f_w(q))\}_{q=0}^{w-1}$ and~$i\in\range{w-1}$, let $C^i_w$ be the cyclic shift of $C_w$ along the $y$-axis, i.e., $C^i_w\triangleq\{(q,(f_w(q)+i)\bmod w\}_{q=0}^{w-1}$. 
    
\end{definition}
Notice that $\bigcup_{0\leq i\leq w-1} C_w^i$ covers all entries of~$\range{w-1}^2$
%~$\{0,1,\ldots,w-1\}^2$ 
since the reversing bit function is bijective. 
Further, it is easy to prove 
\ifFULL\else
(see~\cite{combinedWebsites})
\fi
that shifted VDC sets retain the same bound on the largest empty rectangle.

\begin{lemma}\label{lemma:shifted vdc area}
    For any~$i\in \llbracket w-1\rrbracket$ we have~$A_w(C_w^i)<4w$.
   
\end{lemma}
    \begin{proof}
The proof is identical to the proof of Lemma~\ref{lemma:VDCarea}, since shifting van der Corput set along $y$ axis will not change the bound on $S_y$ and $S_x$. 
Observe that the length of $S_x$ does not change since the shift is along the~$y$-axis, and the bound on~$S_y$ does not change since the distance between the~$y$-coordinates of adjacent points is not affected by the shift.
\end{proof}

The preceding lemmas bound the area of the largest axis-parallel empty rectangle that can occur in a VDC set and in any of its translations parallel to the \(y\)-axis.  
Applying the same argument to translations parallel to the \(x\)-axis yields an identical bound.  
Combining these two cases, we obtain the following statement.
\begin{lemma}\label{lemma: vdc area mutiple shifts}
    Let $C^{i,j}_w$ be the cyclic shift of $C_w$ along the $y$-axis for distance $i$ then shift $C^{i}_w$ along the $x$-axis for distance $j$.
    Then, we have $A_w(C_w^{i,j})<4w$.
\end{lemma}

In the upcoming sections, we will also require tiling of a VDC set~$C_w$ in a larger~$z\times z$ codeword for some~$z$ divisible by~$w$. To bound the largest empty rectangle of this tiling, 
\ifFULL
we have the following lemma.
\else
we have the following lemma, whose proof is given in~\cite{combinedWebsites}.
\fi

\begin{lemma}\label{union of vdc lemma}
    Let~$w$ be power of~$2$ and~$z$ be an integer multiple of~$w$. For~$x,y\in\llbracket z/w-1\rrbracket$ let~$C_w^{(x,y)}\triangleq \{(z_1+xw,z_2+yw)\vert (z_1,z_2)\in C_w\}$, and let~$T\triangleq \bigcup_{x,y\in \llbracket z/w-1\rrbracket}C_{w}^{(x,y)}$. 
    Then, we have~$A_z(T)<4w$.
    %for $[0,z]\times [0,z]$.
\end{lemma}

\begin{proof}
    Let~$B=S_a\times S_b=[\alpha,\beta]\times[\gamma,\eta]\subseteq[0,z]\times[0,z]$ be an empty rectangle, i.e.  $S_a=[\alpha,\beta]$ and~$S_b=[\gamma,\eta]$. 
     Recall that the bit reversing function is bijective and that $T$ is constructed as a union of shifted VDC sets. 
     Therefore, the distance between any two adjacent points of~$T$ on the same row or column is exactly $w$, and hence $|S_a|, |S_b|< w$. 
     To prove the theorem we show that the area of~$B$ is smaller than~$4w$. We split to cases:
     \begin{itemize}
         \item Case 1: $B\subseteq  [xw,(x+1)w]\times [yw,(y+1)w]$ for some~$x,y\in\llbracket z/w-1\rrbracket$. Since $B$ is located inside a shift of $[0,w]\times[0,w]$, the maximum area of an empty axis-parallel rectangle that does not contain a point from $C_w^{(x,y)}$ is $4w$ by Lemma~\ref{lemma:VDCarea}, and hence~$A_z(B)<4w$.
         \item Case 2: $B$ intersects $[xw,(x+1)w]\times [yw,(y+1)w]$ and $ [(x+1)w,(x+2)w]\times [yw,(y+1)w]$ for some~$x,y\in\llbracket z/w-1\rrbracket$. Let $\alpha\in[0,z]\times[0,z]$ be an integer. Points inside $[\alpha,\alpha+w]\times [yw,(y+1)w]$ can be viewed as shift $C_w^{(x,y)}$ to the right by $\alpha-xw$. 
         By Lemma~\ref{lemma:shifted vdc area} it follows that a shift of a VDC set retains the area of the largest empty rectangle inside it, and thus~$A_z(B)<4w$.
         \item Case 3: $B$ intersects $[xw,(x+1)w]\times [yw,(y+1)w]$ and $ [xw,(x+1)w]\times [(y+1)w,(y+2)w]$.
         Similar to Case 2.
         \item Case 4: $B$ intersects $[xw,(x+1)w]\times [yw,(y+1)w]$, $ [xw,(x+1)w]\times [(y+1)w,(y+2)w]$, $ [(x+1)w,(x+2)w]\times [yw,(y+1)w]$ and $[(x+1)w,(x+2)w]\times[(y+1)w,(y+2)w]$.
         Points inside $[\alpha,\alpha+w]\times [\gamma,\gamma+w]$ can be viewed as a shift of $C_w^{(x,y)}$ to the right by $\alpha-xw$ followed by a shift to the top by~$\gamma-yw$.
         By Lemma~\ref{lemma: vdc area mutiple shifts} it follows that a shift of a VDC set retains the area of largest empty rectangle inside it, and thus~$A_z(B)<4w$.\qedhere
     \end{itemize}
\end{proof}

Lemma~\ref{union of vdc lemma} shows that tiling small VDC sets in a large matrix does not significantly affect the area of largest empty rectangle. 
We will also require shifts of the set~$T$, i.e., for~$i\in\llbracket w-1\rrbracket$ let
\begin{align}\label{equation:T_i}
    T_i\triangleq\{(u,(v+i)\bmod z)\vert (u,v)\in T \}.
\end{align}
Since $\bigcup_{0\leq i\leq w-1} C_w^i$ covers every entry in the $w\times w$ matrix, it follows that $\cup_{i\in \llbracket w-1\rrbracket} T_i$ covers all entries in the $z\times z$ matrix. In addition, it follows from previous lemmas that $A_z(T_i)<4w$  for every~$i\in\llbracket w-1\rrbracket$. 

\subsection{Encoding}\label{section:encoding for 2D}
To construct $(q,N,M,h)$-forensic code we wish to disperse the information in a way such that:
\begin{enumerate}
    \item[\textbf{A}.] Every legal fragment (i.e., whose area is at least~$M$ and side length is at least~$h$) contains all information.
    \item[\textbf{B}.] The information in the fragment can be aligned against the original codeword.
\end{enumerate}
To this end, let~$d$ be the largest integer such that~$h \geq 3d-1$, and let~$m$ be the largest integer power of~$2$ such that~$M \geq (4d-1)((m+1)d+d-1)$;
these parameters are chosen to ensure that each fragment will contain at least $m$ units. 
For convenience we also assume that~$m|n$ and~$d|n$, although the construction can be generalized to other cases as well.
%recall the parameters~$d$ and~$m$ defined in Section~\ref{section:overview}. 
%Suppose that~$d|n$, and p
Partition the~$n\times n$ codeword into~$d\times d$ rectangles called \textit{units}.
Then, the~$\frac{n}{d}\times \frac{n}{d}$ grid of units is colored in~$m'\triangleq\frac{m}{4}$ colors. %, where~$m$ was defined in Section~\ref{section:overview}.

To guarantee~\textbf{A}, the information word~$\boldx$ is partitioned to~$m'$ different parts, and each part is associated with one of the~$m'$ colors. 
The~$m'$ parts of the information word are embedded into the codeword according to a coloring function, meaning, the grid of units is colored in~$m'$ colors, and all units colored in color~$i$ contain the~$i$'th information part.
The challenge in devising this coloring scheme is ensuring that every large enough fragment contains a unit of each color.

To guarantee~\textbf{B}, \textit{Zero Square Identification} encoding is introduced, i.e., all units of color~$0$ are set to only contain zeros, while avoiding the all zero square everywhere else.
Zero square identification encoding
(i) writes an all–zero square on every unit of color~$0$, and
(ii) on every unit of color other than~$0$ it sets the \emph{top-left}, \emph{top-right}, \emph{bottom-left} and \emph{bottom-right} entries of the unit to be~$1$.
%For consistency of the section, w
We defer the proof that the only $d\times d$ all-zero squares are precisely the color~$0$ units to Lemma~\ref{lemma:only all zero square are color 0 units}.
 
In detail, let~$\boldx\in\Sigma_q^{k}$ be the message to be encoded. We begin by mapping~$\boldx$ to~$m'-1$ \textit{distinct} strings, each containing~$R=d^2-4$ bits, using an invertible function~$g_1(\boldx)=\{\boldx_i\}_{i=1}^{m'-1}$, where~$k=(m'-1)(R-\log_q(m'-1))$.
The function~$g_1$ is given in Proposition~\ref{proposition:g1}.

\begin{proposition}\label{proposition:g1}
    For $k=(m'-1)(R-\log_q(m'-1))$, the function~$g_1$ is computed in linear time by partitioning~$\boldx$ into~$m'-1$ segments of length~$R-\log_q(m'-1)$ each, and appending to each segment its index in base~$q$.
    It is readily verified that the resulting segments are distinct and that the encoding is easily invertible. 
\end{proposition}
% \begin{proof}
% To construct~$g_1$, partition the input~$\boldx$ into $m'$ segments of length~$R-\log_q(m')$ each, and append to each segment its index in base~$q$. 
% It is readily verified that the resulting segments are distinct and that the encoding is easily invertible. The function~$g_1$ can be computed in polynomial time as as the operations required—partitioning the input, computing and appending indices can all be performed efficiently.
% \end{proof}

Next, we introduce a function $g_2:\Sigma_q^{R}\to\Sigma_q^{d\times d}$, which sets the top-left, top-right, bottom-left and bottom-right corners of the output to~$1$, and the remaining~$R=d^2-4$ bits of the output contain the~$R$ bits of the input.
%$\boldx_i\in \Sigma_q^{d\times d-2}$ is aligned in the rest entries. 
We then let~$\boldX_i\triangleq g_2(\boldx_i)\in\Sigma_q^{d\times d}$ for each~$i\in\{1,\ldots,m'-1\}$, and embed each~$\boldX_i$ into the $n\times n$ codeword as follows. 
We color the~$\frac{n}{d}\times \frac{n}{d}$ grid of units in~$m'$ colors using a coloring function $c:\llbracket n/d-1\rrbracket\times \llbracket n/d-1\rrbracket\to \llbracket m'-1\rrbracket$, defined as
\begin{align}\label{coloring function}
    c(i,j)=(j-f_{m'}(i \bmod m'))\bmod m'
\end{align}
for every~$(i,j)\in \llbracket n/d-1\rrbracket\times \llbracket n/d-1\rrbracket$, where~$f_{m'}$ is the bit-reversal function  (Definition~\ref{definition:Bit reversing function}), 
and place~$\boldX_i$ in all units colored by~$i$ for~$i\in\{1,\ldots,m'-1\}$,  while all entries of all units of color~$0$ are set to~$0$.
That is, the codeword~$C(\boldx)\in\Sigma_q^{n\times n}$ is such that the~$(i,j)$'th $d\times d$ submatrix equals~$\boldX_{c(i,j)}$, for every $(i,j)\in \llbracket n/d-1\rrbracket\times \llbracket n/d-1\rrbracket$.
An example is given in Fig.~\ref{fig:enter-label}.

Now let~$L_0,\ldots,L_{m'-1}\subseteq \llbracket n/d-1\rrbracket\times \llbracket n/d-1\rrbracket$ be the color sets induced by~\eqref{coloring function}, i.e.,~$L_r=\{(i,j)\vert c(i,j)=r\}$. We prove that each~$L_i$ is a union of VDC sets in Lemma~\ref{lemma:valid coloring as vdc set}; this fact will be required for decoding in Section~\ref{section:decoding for 2D}.
\begin{lemma}\label{lemma:valid coloring as vdc set}
    Every color set induced by the coloring function~\eqref{coloring function} is of the form~$T_i$ mentioned in~\eqref{equation:T_i}, i.e. $L_i=T_i$ for all $i\in \llbracket m'-1\rrbracket$.
\end{lemma}
\begin{proof}
    See Appendix~\ref{section:omittedProofs}, Lemma~\ref{app:proof-lemma:valid coloring as vdc set}.
\end{proof}

\begin{lemma}\label{lemma:only all zero square are color 0 units}
    The only $d\times d$ all-zero squares are the units of color~$0$.
\end{lemma}
\begin{proof}
    Recall that zero square identification encoding writes
(i) an all–zero square on every unit of color~$0$, and
(ii) on every unit of color other than $0$ it sets the corner entries of the unit to be~$1$.
We prove the statement by case analysis of different $d\times d$ square.
\begin{itemize}
    \item If the $d\times d$ square lies entirely within a non-$0$ unit, then it contains a $1$ at every corner entry.
    \item If the $d\times d$ square intersects multiple units, then as long as it intersects with a unit of color other than zero then it contains at least a $1$. 
    So, the only case which remains to exclude is that the $d\times d$ square only intersects units of color~$0$.
    With Definition~\ref{definition:Bit reversing function} and~\eqref{coloring function}, two adjacent units cannot share the same color as $f_{m'}$ is bijective, i.e. $c(i+1,j)\neq c(i,j)$ and $c(i,j+1)\neq c(i,j)$.
    Thus, any $d\times d$ square intersects at least a unit of color other than zero. \qedhere
\end{itemize}
    %This ensures that the only $d\times d$ all-zero square are precisely the color~$0$ units.
\end{proof}

\subsection{Decoding}\label{section:decoding for 2D}
%\subsection{Decoding method: Zero Square Free and Zero Square Identification}
%We say that a~$d\times d$ rectangle of a given legal fragment is a \textit{complete unit} if we can identify the boundaries of this unit \red{[Consider if ``complete unit'' is necessary.]}. 
Given a legal fragment taken from~$C(\boldx)$ to decode, the decoder first identifies all units contained in the fragment by aligning them against~$C(\boldx)$.
Since~$C(\boldx)$ is unknown, the decoder identifies the boundaries between all units 
%as follows.
%In zero square identification encoding ,
%\blue{The decoder identifies the boundaries 
by finding 
%the unique \red{[not necessarily unique, but that shouldn't matter]}
a~$d\times d$ zero square in the fragment, i.e., the decoder identifies a unit of color~$0$.
It remains to prove there exists a unit of color~$0$ in any legal fragment, and that any~$d\times d$ rectangle of zeros must be a unit, which is done in Lemma~\ref{lemma:enough complete units} and Lemma~\ref{lemma s=m'} given shortly. 

Having identified the boundaries, the decoder places all units in a set~$\{\boldY_i\}_{i=1}^s$ of some size~$s$ (without repetitions). 
The decoder applies~$g_2^{-1}$ on each~$\boldY_i$, and later applies~$g_1^{-1}$  on the resulting set. 
That is, the output of the decoding algorithm is~$\hat{\boldx}=g_1^{-1}(\{g_2^{-1}(\boldY_i)\}_{i=1}^s)$ .
We now turn to prove that the decoding algorithm is well-defined, that the above alignment is possible, and that~$\hat{\boldx}=\boldx$.

To prove these statements, we show that the fragment contains all~$m'$ distinct units.
This is done in two steps: we first show that the fragment contains sufficiently many units (Lemma~\ref{lemma:enough complete units}) and then show that all~$m'$ \textit{distinct} units are among them (Lemma~\ref{lemma s=m'}).
\ifFULL\else
Both lemmas are proved in~\cite{combinedWebsites}.
\fi

\begin{lemma}\label{lemma:enough complete units}
  A legal fragment contains a grid of~$(x+1)\times (y+1)$ units for some positive integers~$x,y$ such that~$xy=m$.
\end{lemma}

\ifFULL
\begin{proof}
    %By Lemma~\ref{lemma:cubefree}, constructions in Lemma~\ref{lemma:g2} and by the definition of~$C(\boldx)$, it is readily verified that the decoder can identify all complete units in the fragment by finding all~$\ell\times \ell$ zero squares.
    %It remains to show that the fragment contains sufficiently many complete units.
    We prove the statement by case analysis for all different side lengths of a legal fragment. 
    Suppose the fragment is of dimensions~$a\times b$ for some integers~$a,b$ such that~$ab\ge (4d-1)((m+1)d+d-1)$ (as assumed in Section~\ref{section:encoding for 2D}).
    \begin{itemize}
    \item If~$a<3d-1$ the constraint~$h\ge 3d-1$ is violated, and hence such fragments are not considered legal.
    \item If $3d-1\le a<4d-1$, we have $b \ge (m+1)d+d-1$. Hence, it is clear that the fragment contains a~$u\times v$ rectangle of units for~$(u,v)=(2d,(m+1)d)$.
    This pair~$(u,v)$ can be represented as~$u=(x+1)d$ and~$v=(y+1)d$ for~$(x,y)=(1,m)$, and indeed~$xy=m$.

    \item If~$ 4d-1\le a \le (m+1)d+d-1$, recall that $m=2^p$, and let~$j\in \{1,2, \ldots, p-1\}$ be an integer such that $(2^{j} + 1)d + d - 1\le a <(2^{j+1}+1)d + d - 1$. 
 
    In Lemma~\ref{Lemma: forInequlaity} (see Appendix~\ref{section:omittedProofs}) it is shown that
    \begin{align*}
        ((2^{i}+1) d  + d - 1)((2^{p+1-i}+1) d + d - 1) \le (4d-1)((2^p+1)d+d-1),
    \end{align*}
    for every integer~$2\le i\le p-1$. Hence, it follows that

    \begin{align*}
        b &\geq \frac{(4d - 1)((2^p+1) d + d - 1)}{(2^{j+1}+1)d + d - 1} \\
        &> \frac{((2^{j+1}+1) d + d - 1)((2^{p-j}+1) d + d - 1)}{(2^{j+1}+1) d + d - 1} \\
        &= (2^{p-j}+1) d + d - 1.
    \end{align*}
    Therefore, it is evident that the fragment contains a~$u\times v$ rectangle for~$(u,v)=((2^j+1)d,(2^{p-j}+1)d)$ and hence~$(x,y)=(2^j,2^{p-j})$, which satisfy $xy=m$.
    \item If $a> (m+1)d+d-1$, we further separate $b\le 4d-1$ and $b>4d-1$
    \begin{itemize}
        \item If $b\le 4d-1$,  it is clear that the fragment contains a~$u\times v$ rectangle of  units for~$(u,v)=((m+1)d,2d)$ and hence $xy=m$ for $(x,y)=(1,m)$.
        \item If $b>4d-1$, the fragment contains more than $4d-1$ columns, and we remove extra columns so that exactly $4d-1$ columns are left.
        The resulting rectangle is of dimensions $a\times b$ with $a>(m+1)d+d+1$ and $b=4d-1$ which is the condition for the previous case. \qedhere
    \end{itemize}
    %There is at least $2^{p-j} + 1$ complete units along $y$-axis and thus $(2^j + 1 - 1) \cdot (2^{p-j} + 1 - 1) = m$ satisfies the requirement.
    \end{itemize}
\end{proof}
\fi

\begin{lemma}\label{lemma s=m'}
    In the above decoding process we have~$s=m'$.
\end{lemma}
\ifFULL
\begin{proof}
%We wish to prove that the fragment contains all~$m'$ distinct  units.
%To this end, w
We show that the rectangle whose existence is guaranteed by Lemma~\ref{lemma:enough complete units} contains all~$m'$ distinct  units.
Let~$P$ be the $(x+1)(y+1)$ grid of  units guaranteed by Lemma~\ref{lemma:enough complete units}. Define a $(x+1)(y+1)$ grid by contracting the~$(i,j)$'the  unit to the point $(i,j)$. 
With this contracting process, the area of the contracted grid $P$ is exactly~$xy=m$ (see example in Fig. \ref{fig:enter-label}).

% Color the resulting grid so that every every two contracted points are colored in the same color if and only if their respective complete units are identical.
% Therefore, the colors in the contracted grid are identical, up to permutation, to the ones given by the coloring function in~\eqref{coloring function}.
It suffices to show that the contracted grid contains all colors given by the coloring function in~\eqref{coloring function}. 
To this end, recall that~$L_0,\ldots,L_{m'-1}\subseteq \llbracket n/d-1\rrbracket\times \llbracket n/d-1\rrbracket$ are the color sets induced by~\eqref{coloring function} from Lemma~\ref{lemma:valid coloring as vdc set}. 
Suppose there exists some color $a$ which is not in the rectangle in Lemma~\ref{lemma:enough complete units}. 
Lemma~\ref{lemma:valid coloring as vdc set} shows that $L_a=T_a$, where~$T_a$ is a shifted VDC set. 
%Thus, we can apply Lemma~\ref{union of vdc lemma} which gives us the following property: 
Thus, Lemma~\ref{union of vdc lemma} implies that the area of any axis-parallel rectangle which does not contain color~$a$ is strictly smaller than~$4m'$. 
Since the area of $P$ is exactly $m=4m'$, it follows that it contains the color $a$, a contradiction.
\end{proof}
\fi

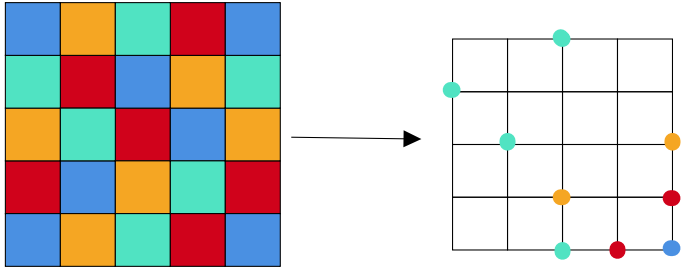
\begin{figure}
    \centering    \begin{tikzpicture}[x=0.75pt,y=0.75pt,yscale=-1,xscale=1]

%Shape: Rectangle [id:dp7907133418036909] 
\draw  [fill={rgb, 255:red, 208; green, 2; blue, 27 }  ,fill opacity=1 ] (99,216.94) -- (126.45,216.94) -- (126.45,243.47) -- (99,243.47) -- cycle ;
%Shape: Rectangle [id:dp6949146710528247] 
\draw  [fill={rgb, 255:red, 245; green, 166; blue, 35 }  ,fill opacity=1 ] (99,190.41) -- (126.45,190.41) -- (126.45,216.94) -- (99,216.94) -- cycle ;
%Shape: Rectangle [id:dp4921853572418293] 
\draw  [fill={rgb, 255:red, 74; green, 144; blue, 226 }  ,fill opacity=1 ] (126.45,216.94) -- (153.89,216.94) -- (153.89,243.47) -- (126.45,243.47) -- cycle ;
%Shape: Rectangle [id:dp5690907676284465] 
\draw  [fill={rgb, 255:red, 208; green, 2; blue, 27 }  ,fill opacity=1 ] (126.45,163.88) -- (153.89,163.88) -- (153.89,190.41) -- (126.45,190.41) -- cycle ;
%Shape: Rectangle [id:dp7530566661591171] 
\draw  [fill={rgb, 255:red, 208; green, 2; blue, 27 }  ,fill opacity=1 ] (153.89,190.41) -- (181.34,190.41) -- (181.34,216.94) -- (153.89,216.94) -- cycle ;
%Shape: Rectangle [id:dp5900255173555558] 
\draw  [fill={rgb, 255:red, 74; green, 144; blue, 226 }  ,fill opacity=1 ] (181.34,190.41) -- (208.78,190.41) -- (208.78,216.94) -- (181.34,216.94) -- cycle ;
%Shape: Rectangle [id:dp4862214530486586] 
\draw  [fill={rgb, 255:red, 74; green, 144; blue, 226 }  ,fill opacity=1 ] (153.89,163.88) -- (181.34,163.88) -- (181.34,190.41) -- (153.89,190.41) -- cycle ;
%Shape: Rectangle [id:dp6039094263311469] 
\draw  [fill={rgb, 255:red, 80; green, 227; blue, 194 }  ,fill opacity=1 ] (126.45,190.41) -- (153.89,190.41) -- (153.89,216.94) -- (126.45,216.94) -- cycle ;
%Shape: Rectangle [id:dp870830323572616] 
\draw  [fill={rgb, 255:red, 245; green, 166; blue, 35 }  ,fill opacity=1 ] (153.89,216.94) -- (181.34,216.94) -- (181.34,243.47) -- (153.89,243.47) -- cycle ;
%Shape: Rectangle [id:dp6934349490104768] 
\draw  [fill={rgb, 255:red, 245; green, 166; blue, 35 }  ,fill opacity=1 ] (181.34,163.88) -- (208.78,163.88) -- (208.78,190.41) -- (181.34,190.41) -- cycle ;
%Shape: Rectangle [id:dp4666217677858966] 
\draw  [fill={rgb, 255:red, 80; green, 227; blue, 194 }  ,fill opacity=1 ] (181.34,216.94) -- (208.78,216.94) -- (208.78,243.47) -- (181.34,243.47) -- cycle ;
%Shape: Rectangle [id:dp5947258913945603] 
\draw  [fill={rgb, 255:red, 80; green, 227; blue, 194 }  ,fill opacity=1 ] (99,163.88) -- (126.45,163.88) -- (126.45,190.41) -- (99,190.41) -- cycle ;
%Shape: Rectangle [id:dp9957487665276485] 
\draw  [fill={rgb, 255:red, 208; green, 2; blue, 27 }  ,fill opacity=1 ] (181.34,137.35) -- (208.78,137.35) -- (208.78,163.88) -- (181.34,163.88) -- cycle ;
%Shape: Rectangle [id:dp06472390437395692] 
\draw  [fill={rgb, 255:red, 74; green, 144; blue, 226 }  ,fill opacity=1 ] (99,137.35) -- (126.45,137.35) -- (126.45,163.88) -- (99,163.88) -- cycle ;
%Shape: Rectangle [id:dp03728269410923857] 
\draw  [fill={rgb, 255:red, 245; green, 166; blue, 35 }  ,fill opacity=1 ] (126.45,137.35) -- (153.89,137.35) -- (153.89,163.88) -- (126.45,163.88) -- cycle ;
%Shape: Rectangle [id:dp4504548806215458] 
\draw  [fill={rgb, 255:red, 80; green, 227; blue, 194 }  ,fill opacity=1 ] (153.89,137.35) -- (181.34,137.35) -- (181.34,163.88) -- (153.89,163.88) -- cycle ;
%Shape: Rectangle [id:dp5636293788273272] 
\draw  [fill={rgb, 255:red, 208; green, 2; blue, 27 }  ,fill opacity=1 ] (208.78,216.94) -- (236.23,216.94) -- (236.23,243.47) -- (208.78,243.47) -- cycle ;
%Shape: Rectangle [id:dp9153035922393451] 
\draw  [fill={rgb, 255:red, 245; green, 166; blue, 35 }  ,fill opacity=1 ] (208.78,190.41) -- (236.23,190.41) -- (236.23,216.94) -- (208.78,216.94) -- cycle ;
%Shape: Rectangle [id:dp9509620114283157] 
\draw  [fill={rgb, 255:red, 80; green, 227; blue, 194 }  ,fill opacity=1 ] (208.78,163.88) -- (236.23,163.88) -- (236.23,190.41) -- (208.78,190.41) -- cycle ;
%Shape: Rectangle [id:dp21842529568287672] 
\draw  [fill={rgb, 255:red, 74; green, 144; blue, 226 }  ,fill opacity=1 ] (208.78,137.35) -- (236.23,137.35) -- (236.23,163.88) -- (208.78,163.88) -- cycle ;
%Shape: Rectangle [id:dp08783695981513984] 
\draw  [fill={rgb, 255:red, 74; green, 144; blue, 226 }  ,fill opacity=1 ] (99,243.47) -- (126.45,243.47) -- (126.45,270) -- (99,270) -- cycle ;
%Shape: Rectangle [id:dp47170369016977176] 
\draw  [fill={rgb, 255:red, 245; green, 166; blue, 35 }  ,fill opacity=1 ] (126.45,243.47) -- (153.89,243.47) -- (153.89,270) -- (126.45,270) -- cycle ;
%Shape: Rectangle [id:dp05792467315399508] 
\draw  [fill={rgb, 255:red, 80; green, 227; blue, 194 }  ,fill opacity=1 ] (153.89,243.47) -- (181.34,243.47) -- (181.34,270) -- (153.89,270) -- cycle ;
%Shape: Rectangle [id:dp3118198186982897] 
\draw  [fill={rgb, 255:red, 208; green, 2; blue, 27 }  ,fill opacity=1 ] (181.34,243.47) -- (208.78,243.47) -- (208.78,270) -- (181.34,270) -- cycle ;
%Shape: Rectangle [id:dp233327207497942] 
\draw  [fill={rgb, 255:red, 74; green, 144; blue, 226 }  ,fill opacity=1 ] (208.78,243.47) -- (236.23,243.47) -- (236.23,270) -- (208.78,270) -- cycle ;
%Straight Lines [id:da03525322556078159] 
\draw    (241.71,205.05) -- (303.67,205.92) ;
\draw [shift={(306.67,205.96)}, rotate = 180.81] [fill={rgb, 255:red, 0; green, 0; blue, 0 }  ][line width=0.08]  [draw opacity=0] (8.93,-4.29) -- (0,0) -- (8.93,4.29) -- cycle    ;
%Shape: Rectangle [id:dp06703843343070348] 
\draw   (322.22,155.65) -- (349.66,155.65) -- (349.66,182.18) -- (322.22,182.18) -- cycle ;
%Shape: Rectangle [id:dp774492248154552] 
\draw   (377.11,155.65) -- (404.55,155.65) -- (404.55,182.18) -- (377.11,182.18) -- cycle ;
%Shape: Rectangle [id:dp616481338972583] 
\draw   (322.22,182.18) -- (349.66,182.18) -- (349.66,208.71) -- (322.22,208.71) -- cycle ;
%Shape: Rectangle [id:dp1633926297896493] 
\draw   (349.66,155.65) -- (377.11,155.65) -- (377.11,182.18) -- (349.66,182.18) -- cycle ;
%Shape: Rectangle [id:dp6827578045643969] 
\draw   (349.66,182.18) -- (377.11,182.18) -- (377.11,208.71) -- (349.66,208.71) -- cycle ;
%Shape: Rectangle [id:dp7356732167690547] 
\draw   (377.11,182.18) -- (404.55,182.18) -- (404.55,208.71) -- (377.11,208.71) -- cycle ;
%Shape: Rectangle [id:dp8441935528557825] 
\draw   (404.55,155.65) -- (432,155.65) -- (432,182.18) -- (404.55,182.18) -- cycle ;
%Shape: Rectangle [id:dp6414515676408175] 
\draw   (322.22,208.71) -- (349.66,208.71) -- (349.66,235.24) -- (322.22,235.24) -- cycle ;
%Shape: Rectangle [id:dp07639961673323326] 
\draw   (322.22,235.24) -- (349.66,235.24) -- (349.66,261.77) -- (322.22,261.77) -- cycle ;
%Shape: Rectangle [id:dp5982525565769439] 
\draw   (349.66,208.71) -- (377.11,208.71) -- (377.11,235.24) -- (349.66,235.24) -- cycle ;
%Shape: Rectangle [id:dp2401659152569564] 
\draw   (349.66,235.24) -- (377.11,235.24) -- (377.11,261.77) -- (349.66,261.77) -- cycle ;
%Shape: Rectangle [id:dp30134247684329374] 
\draw   (377.11,208.71) -- (404.55,208.71) -- (404.55,235.24) -- (377.11,235.24) -- cycle ;
%Shape: Rectangle [id:dp3390002430383332] 
\draw   (377.11,235.24) -- (404.55,235.24) -- (404.55,261.77) -- (377.11,261.77) -- cycle ;
%Shape: Rectangle [id:dp17740991270476103] 
\draw   (404.55,182.18) -- (432,182.18) -- (432,208.71) -- (404.55,208.71) -- cycle ;
%Shape: Rectangle [id:dp38334430451666734] 
\draw   (404.55,208.71) -- (432,208.71) -- (432,235.24) -- (404.55,235.24) -- cycle ;
%Shape: Rectangle [id:dp6473321816880442] 
\draw   (404.55,235.24) -- (432,235.24) -- (432,261.77) -- (404.55,261.77) -- cycle ;
%Shape: Free Drawing [id:dp39951963807809854] 
\draw  [color={rgb, 255:red, 74; green, 144; blue, 226 }  ,draw opacity=1 ][line width=6] [line join = round][line cap = round] (322.22,155.65) .. controls (322.22,155.65) and (322.22,155.65) .. (322.22,155.65) ;
%Shape: Free Drawing [id:dp616948529149026] 
\draw  [color={rgb, 255:red, 74; green, 144; blue, 226 }  ,draw opacity=1 ][line width=6] [line join = round][line cap = round] (432,154.73) .. controls (432,154.73) and (432,154.73) .. (432,154.73) ;
%Shape: Free Drawing [id:dp6268313358573656] 
\draw  [color={rgb, 255:red, 74; green, 144; blue, 226 }  ,draw opacity=1 ][line width=6] [line join = round][line cap = round] (377.11,181.26) .. controls (377.11,181.26) and (377.11,181.26) .. (377.11,181.26) ;
%Shape: Free Drawing [id:dp3066274467699799] 
\draw  [color={rgb, 255:red, 74; green, 144; blue, 226 }  ,draw opacity=1 ][line width=6] [line join = round][line cap = round] (405.47,207.79) .. controls (405.47,207.79) and (405.47,207.79) .. (405.47,207.79) ;
%Shape: Free Drawing [id:dp1988916233783511] 
\draw  [color={rgb, 255:red, 74; green, 144; blue, 226 }  ,draw opacity=1 ][line width=6] [line join = round][line cap = round] (348.75,234.32) .. controls (348.75,234.32) and (348.75,234.32) .. (348.75,234.32) ;
%Shape: Free Drawing [id:dp49624562431636643] 
\draw  [color={rgb, 255:red, 74; green, 144; blue, 226 }  ,draw opacity=1 ][line width=6] [line join = round][line cap = round] (322.22,260.85) .. controls (322.22,260.85) and (322.22,260.85) .. (322.22,260.85) ;
%Shape: Free Drawing [id:dp27861488363137177] 
\draw  [color={rgb, 255:red, 74; green, 144; blue, 226 }  ,draw opacity=1 ][line width=6] [line join = round][line cap = round] (432,260.85) .. controls (431.7,260.85) and (431.39,260.85) .. (431.09,260.85) ;
%Shape: Free Drawing [id:dp22189617184010157] 
\draw  [color={rgb, 255:red, 245; green, 166; blue, 35 }  ,draw opacity=1 ][line width=6] [line join = round][line cap = round] (349.66,155.65) .. controls (349.66,155.65) and (349.66,155.65) .. (349.66,155.65) ;
%Shape: Free Drawing [id:dp7329668720466767] 
\draw  [color={rgb, 255:red, 245; green, 166; blue, 35 }  ,draw opacity=1 ][line width=6] [line join = round][line cap = round] (404.55,182.18) .. controls (404.55,182.18) and (404.55,182.18) .. (404.55,182.18) ;
%Shape: Free Drawing [id:dp8790562611030937] 
\draw  [color={rgb, 255:red, 245; green, 166; blue, 35 }  ,draw opacity=1 ][line width=6] [line join = round][line cap = round] (432,207.79) .. controls (432,207.49) and (432,207.18) .. (432,206.88) ;
%Shape: Free Drawing [id:dp47075416248314905] 
\draw  [color={rgb, 255:red, 245; green, 166; blue, 35 }  ,draw opacity=1 ][line width=6] [line join = round][line cap = round] (377.11,235.24) .. controls (376.8,235.24) and (376.5,235.24) .. (376.2,235.24) ;
%Shape: Free Drawing [id:dp07602268969089665] 
\draw  [color={rgb, 255:red, 245; green, 166; blue, 35 }  ,draw opacity=1 ][line width=6] [line join = round][line cap = round] (348.75,260.85) .. controls (348.75,260.85) and (348.75,260.85) .. (348.75,260.85) ;
%Shape: Free Drawing [id:dp7620367598165385] 
\draw  [color={rgb, 255:red, 80; green, 227; blue, 194 }  ,draw opacity=1 ][line width=6] [line join = round][line cap = round] (377.11,155.65) .. controls (376.68,155.65) and (376.2,155.16) .. (376.2,154.73) ;
%Shape: Free Drawing [id:dp4067797326719158] 
\draw  [color={rgb, 255:red, 80; green, 227; blue, 194 }  ,draw opacity=1 ][line width=6] [line join = round][line cap = round] (432,181.26) .. controls (432,181.26) and (432,181.26) .. (432,181.26) ;
%Shape: Free Drawing [id:dp3892284373619943] 
\draw  [color={rgb, 255:red, 80; green, 227; blue, 194 }  ,draw opacity=1 ][line width=6] [line join = round][line cap = round] (322.22,181.26) .. controls (321.91,181.26) and (321.61,181.26) .. (321.3,181.26) ;
%Shape: Free Drawing [id:dp7838625624681066] 
\draw  [color={rgb, 255:red, 80; green, 227; blue, 194 }  ,draw opacity=1 ][line width=6] [line join = round][line cap = round] (349.66,207.79) .. controls (349.66,207.49) and (349.66,207.18) .. (349.66,206.88) ;
%Shape: Free Drawing [id:dp12320169271924031] 
\draw  [color={rgb, 255:red, 80; green, 227; blue, 194 }  ,draw opacity=1 ][line width=6] [line join = round][line cap = round] (404.55,235.24) .. controls (404.55,235.24) and (404.55,235.24) .. (404.55,235.24) ;
%Shape: Free Drawing [id:dp029554464246024592] 
\draw  [color={rgb, 255:red, 80; green, 227; blue, 194 }  ,draw opacity=1 ][line width=6] [line join = round][line cap = round] (377.11,261.77) .. controls (377.11,262.07) and (377.11,262.38) .. (377.11,262.68) ;
%Shape: Free Drawing [id:dp9532129671989302] 
\draw  [color={rgb, 255:red, 208; green, 2; blue, 27 }  ,draw opacity=1 ][line width=6] [line join = round][line cap = round] (404.55,155.19) .. controls (404.55,155.19) and (404.55,155.19) .. (404.55,155.19) ;
%Shape: Free Drawing [id:dp20900842830593547] 
\draw  [color={rgb, 255:red, 208; green, 2; blue, 27 }  ,draw opacity=1 ][line width=6] [line join = round][line cap = round] (349.66,180.8) .. controls (349.66,180.8) and (349.66,180.8) .. (349.66,180.8) ;
%Shape: Free Drawing [id:dp3359856119042752] 
\draw  [color={rgb, 255:red, 208; green, 2; blue, 27 }  ,draw opacity=1 ][line width=6] [line join = round][line cap = round] (322.22,234.78) .. controls (322.22,234.78) and (322.22,234.78) .. (322.22,234.78) ;
%Shape: Free Drawing [id:dp5523359695933341] 
\draw  [color={rgb, 255:red, 208; green, 2; blue, 27 }  ,draw opacity=1 ][line width=6] [line join = round][line cap = round] (431.09,235.69) .. controls (431.39,235.69) and (431.7,235.69) .. (432,235.69) ;
%Shape: Free Drawing [id:dp5053251209798604] 
\draw  [color={rgb, 255:red, 208; green, 2; blue, 27 }  ,draw opacity=1 ][line width=6] [line join = round][line cap = round] (404.55,261.31) .. controls (404.55,261.61) and (404.55,261.92) .. (404.55,262.22) ;
%Shape: Free Drawing [id:dp49892621650492464] 
\draw  [color={rgb, 255:red, 245; green, 166; blue, 35 }  ,draw opacity=1 ][line width=6] [line join = round][line cap = round] (322.22,208.25) .. controls (322.22,208.25) and (322.22,208.25) .. (322.22,208.25) ;

%Shape: Free Drawing [id:dp4515091476490656] 
\draw  [color={rgb, 255:red, 208; green, 2; blue, 27 }  ,draw opacity=1 ][line width=6] [line join = round][line cap = round] (377,209) .. controls (377,209) and (377,209) .. (377,209) ;

% Text Node
%\draw (155.89,273) node [anchor=north west][inner sep=0.75pt]   [align=left] {Figure 2: How units are contracted};
\end{tikzpicture}
    \caption{Units in an~$n\times n$ codeword (left) are associated with a colored grid (right), in which every color set is a VDC set.
    %How units are contracted \blue{and how units are colored using shifted VDC sets. }\red{[This figure doesn't seem to be referenced anywhere in the text. Please find an appropriate place in the text and add a reference. Also, note that it exemplifies not only the contraction of units, but also VDC sets and its shifts, so consider adding a reference there as well. ]}
    }
    \label{fig:enter-label}
\end{figure}
\tikzset{every picture/.style={line width=0.75pt}} %set default line width to 0.75pt        
Lemma~\ref{lemma s=m'} shows that the given legal fragment contains all colors, and hence all information.
In particular, legal fragments also contain a unit of color~$0$, which is necessary for alignment in zero square identification encoding.
%\red{[STOPPED HERE Nov 5th: Left to do the last paragraph here, the comparison section, and all figures (first by hand).]}
Therefore, we reverse the encoding operation for the purpose of decoding as shown in the next theorem which conclude the decoding operation.
Resulting rates relative to the theoretical optimum are discussed in Section~\ref{Bound and rates}.
\ifFULL\else
A proof is given in~\cite{combinedWebsites}.
\fi
\begin{proposition}\label{lemma:g_2^-1}
    The functions $g_2^{-1},g_1^{-1}$ can be computed in polynomial time:
    Given a unit $\boldY_i\in\Sigma_q^{d\times d}$, define $g_2^{-1}(\boldY_i)$ as removing the top-left, top-right, bottom-left and bottom-right corners of $1$.
    Finally, applying~$g_1^{-1}$  amounts to removing the first~$\log_q(m'-1)$ indexing bits.
\end{proposition}

\subsection{Bounds and Rates}\label{Bound and rates}
Recall that the rate of a forensic code is~$k/M$.
In what follows we compute the asymptotic rate of the forensic code in Section~\ref{section:encoding for 2D}, and later show that in a different parameter regime there exist code of rate~$1$ (asymptotically). 
%\red{[Add a graph or a table that exemplifies how close are we to the asymptote in finite values.]}
%Our encoding reaches a non-vanishing rate of $\frac{1}{16}$.
\begin{theorem}\label{theorem:comparison of both methods} 
    In a parameter regime where $M\rightarrow \infty$, $h=o(\sqrt{M})$, and $h = \omega(\sqrt{\log_q M})$,
    the rate of the~$(q,N,M,h)$-forensic code in Section~\ref{section:encoding for 2D} approaches~$\frac{1}{32}$.
    %when $M\xrightarrow{} \infty$ with $h=o(\sqrt{M})$ and $h = \omega(\sqrt{\log_q m})$.
\end{theorem}
\begin{proof}
    Recall that for any~$N=n^2$, any~$M\le N$, and any~$h$, we let~$d$ be the largest integer such that~$h\ge 3d-1$, and let~$m$ be the largest power of two such that~$M\ge (4d-1)((m+1)d+d-1)$. 
    Therefore, 
    \begin{equation}\label{equation: bound for M}
        (4d-1)((m+1)d+d-1)\le M\le (4d-1)((2m+1)d+d-1)
    \end{equation}
    and hence
    \begin{align}\label{equation:mUpperLower}
        \frac{1}{2d}\left(\frac{M}{4d-1}+1\right)-1\le m \le \frac{1}{d}\left(\frac{M}{4d-1}+1\right)-1.
    \end{align}
    Since~$d=\Theta(h)$, and since $h=o(\sqrt{M})$, it follows that
    % As $h$ grows faster than $\sqrt{\log_q(m)}$ slower than $\sqrt{M}$, we have $d$ grows faster than $\sqrt{\log_q(m)}$ slower than $\sqrt{M}$.
    % Thus, we have
    \begin{align}\label{equation:Mlimit}
        \lim_{M\rightarrow\infty}\frac{1} {2d}\left(\frac{M}{4d-1}+1\right)= \infty.
    \end{align}
    Therefore, it follows from~\eqref{equation:mUpperLower} and from~\eqref{equation:Mlimit} that in this parameter regime we have $m\xrightarrow[]{}\infty$.
    Further, since~$d=\Theta(h)$, since $h = \omega(\sqrt{\log_q M})$, and since~$M\ge m$, it follows that $$0\le \lim_{M\rightarrow\infty}\frac{\log_q(m)}{d^2}\le \lim_{M\to \infty}\frac{\log_q(M)}{d^2}= 0,$$ 
    therefore $\lim_{M\rightarrow\infty}\frac{\log_q(m)}{d^2}=0$.

    Now, recall that zero square identification encoding requires input of length~$k=(m'-1)(d^2-2-\log_q(m'-1))$ (since it  divides the input into~$m'-1$ parts of length~$d^2-2-\log(m'-1)$ each, where~$m'=\frac{m}{4}$). 
    Therefore, the resulting rate is
    % Then, the function~$g_1$ is used to add unique labels to each part, and the function~$g_2$ is used to align each part into a unit. 
    % Lastly, the coloring function~$c(i,j)$ from~\eqref{coloring function} is used to put the units into the corresponding location of the~$n\times n$ matrix and conclude the encoding.
    % Thus, this encoding has rate 
    $$\frac{(m/4-1)(d^2-4-\log_q(m/4-1))}{M}.$$
    With~\eqref{equation: bound for M}, the rate is at least 
    \begin{align}\label{equation:rateBoundTheorem}
        \frac{(m/4-1)(d^2-4-\log_q(m/4-1))}{(4d-1)(2md+2d-1)}
    \end{align}
    which equal
    $$\frac{1-4/d^2-4/m-\log_q(m/4-1)/d^2+16/(md^2)+4\log_q(m/4-1)/(md^2)}{32+32/m+4/(md^2)-8/d-24/md}.$$
    Thus, the rate is at least $$\frac{1-4/d^2-4/m-\log_q(m/4-1)/d^2}{32+32/m+4/(md^2)},$$
    and since $m,d\overset{M\to\infty}{\to}\infty$ and $\log_q(m)/d^2\xrightarrow[]{}0$, it follows that the rate approaches $\frac{1}{32}$.
\end{proof}
\begin{table}[h]
    \centering
    \caption{Example values of $M,h,d,m$ and the resulting code rate (converging toward $1/32$).}\label{table:someNumbers}
    \begin{tabular}{r r r r r}
        \toprule
        $M$ & $h$ & $d$ & $m$ & \textbf{rate} \\
        \midrule
        1\,024      & 14  & 5   & 8   & 0.022461 \\
        2\,048      & 14  & 5   & 16  & 0.031370 \\
        8\,192      & 26  & 9   & 16  & 0.028350 \\
        16\,384     & 11  & 4   & 256 & 0.030849 \\
        65\,536     & 14  & 5   & 512 & 0.031028 \\
        262\,144    & 155 & 52  & 16  & 0.030904 \\
        1\,048\,576 & 68  & 23  & 256 & 0.031304 \\
        4\,194\,304 & 626 & 209 & 16  & 0.031241 \\
        \bottomrule
    \end{tabular}
\end{table}

\begin{remark}
   When~$M$ and~$h$ are such that $M=(4d-1)((m+1)d+d-1)$ and~$h=3d-1$ for some integers~$d$ and~$m$, a proof similar to that of Theorem~\ref{theorem:comparison of both methods} shows that the rate of the~$(q,N,M,h)$-forensic code in Section~\ref{section:encoding for 2D} approaches~$\frac{1}{16}$; this follows since~$M$ would be equal to its lower bound in~\eqref{equation: bound for M} and since~$m$ would be equal to its upper bound in~\eqref{equation:mUpperLower}, and hence a more accurate bound on the rate would be achieved in~\eqref{equation:rateBoundTheorem}.
   %If $M=(4d-1)((m+1)d+d-1)$ and $h=3d-1$, then the rate converges to $\frac{1}{16}$ and 
   Table~\ref{table:someNumbers for 1/16} presents some examples.
\end{remark}
\begin{table}[h]
\centering
\caption{Example values of $M,h,d,m$, and the resulting code rate (converging toward $1/16$).}\label{table:someNumbers for 1/16}
\begin{tabular}{r r r r r}
\toprule
$M$ & $h$ & $d$ & $m$ & \textbf{rate} \\
\midrule
100{,}045      & 41  & 14  & 128    & 0.057958 \\
1{,}000{,}825  & 131 & 44  & 128    & 0.059689 \\
9{,}973{,}111  & 104 & 35  & 2{,}048  & 0.062100 \\
99{,}053{,}215 & 116 & 39  & 16{,}384 & 0.062219 \\
971{,}469{,}531& 182 & 61  & 65{,}536 & 0.062448 \\
\bottomrule
\end{tabular}
\end{table}

Recall that~$h$ is an artificial restriction which bounds the minimum side length of the axis-parallel rectangle of area~$M$ within the fragment. 
As such, one would like to have~$h$ as small as possible, while simultaneously enabling high rates.
Theorem~\ref{theorem:comparison of both methods} shows that a rather modest restriction on~$h$, i.e., having~$h$ asymptotically larger than~$\sqrt{\log_qM}$, suffices for non-vanishing rates.
Example parameters are given in Table~\ref{table:someNumbers}, showing rapid convergence to~$1/32=0.03125$.

We now turn to show that better codes exist in some narrower parameter regimes. 
Observe that in the forensic code of Section~\ref{section:encoding for 2D}, $n$ can grow arbitrarily large without increasing~$M$ as a result. 
For this parameter regime, it remains open if better constructions exist.
However, by restricting the growth on~$n$ with respect to~$M$, it is possible to show that forensic codes of asymptotic rate~$1$ exist.

\begin{corollary}\label{observation: existence bounds}
    %There exists a~$(q,n,M,h)$-forensic code with rate~$1$ when $\log(M^{1.5}n)\ll M$ by random encoding. 
    If~$\frac{\log(n)}{M}\overset{n\to\infty}{\longrightarrow}0$ then there exists a~$(q,n,M,h)$-forensic code of rate~$1$.
\end{corollary}
\begin{proof}
    Follows by setting~$\delta=0$ in Theorem~\ref{thm: random encoding} in Section~\ref{section:bounds and rates for ecc} which follows.
\end{proof}

\section{Forensic coding for three dimensional information}\label{section: pre for 3D}
\subsection{Problem Definition}\label{section: Problem definition for 3D}
Previous sections discuss matrix codewords, 
under the understanding that to make practical use of our codes, one can embed (2D) matrices in a 3D printed object.
However, this approach clearly overlooks the third dimension, which could be utilized for greater information density.
%to 3D printing process but it will lose information density. 
Therefore, in this section we explore a coding setting in which each codeword is a~3D cuboid (or ``cube''). 
%Instead of codeword being a matrix $C\in \Sigma_q^{n\times n}$ in Section~\ref{section:problem definition for 2D}, a codeword in 3D setting will be a cuboid. 
Specifically, a codeword in this setting is $C\in \Sigma_q^{n\times n\times n'}$ with detailed requirements on $n,n'$ illustrated later; we choose the codeword dimensions as $n\times n\times n'$ for simplicity of exposition, even though other dimensions can be accommodated as well with few technical changes. 
As in the 2D case, we assume that the received fragment is not too thin in any dimension. 
%We assume that the received fragment contains a cuboid of volume $M$ which is neither too thin nor too thick. 
Specifically, a fragment is called \textit{legal} if it contains an~$\alpha\times \beta\times \gamma$ cuboid for some integers $\alpha,\beta,\gamma$ such that $\alpha\beta\gamma\geq M$ and $\min\{\alpha,\beta,\gamma\}\geq h$ where $h$ is an additional parameter.

We construct a forensic code using a discrepancy theoretic notion called \textit{Halton-Hammersely (HH) set} (Definition~\ref{definition: Normalized Halton-Hammersely set} below) which achieves a non-vanishing rate.
Similar to Section~\ref{Section:pre for 2D}, we begin by partitioning the $n\times n \times n'$ codeword into smaller $d\times d\times d$ subcubes called units. 
The received fragment must contain sufficiently many units, and to solve this problem, we let~$d$ be the largest integer such that $h\geq 3d-1$.
Let $c$ be the largest integer such that $a=2^c$ and $b=3^{\lceil c\log_3(2) \rceil}$ satisfy
\begin{align}\label{equation:legal3D}
    M\geq (3d-1)^2((ab+1)d+d-1).
\end{align}
%where $a=2^c$ and $b=3^{\lceil c\log_3(2) \rceil}$; t
These parameters are chosen to ensure that each legal fragment will contain at least $ab$ complete units.
For simplicity, we assume that $ad|n$ and $bd|n'$, but the construction can be generalized to other cases as well.
\begin{comment}
\begin{remark}
    A more careful analysis can lower the constant in the bound in $M\geq (3d-1)^2((ab+1)d+d-1)$ \red{[Unbalanced parentheses, and also this seems identical to~\eqref{equation:legal3D}.]}, and yet the details are omitted for brevity.
\end{remark}
\end{comment}
The encoding process is similar to Section~\ref{section:encoding for 2D}---we assign different units with different colors, while units of the same color contain the same information.
The coloring procedure relies on the properties of the HH set, so that any legal fragment will contain units of all colors, and hence the message can be decoded as it contains all different types of units.
The following section introduces HH sets, as well as necessary discrepancy theoretic theorems. 

\subsection{Halton-Hammersely set}
\begin{definition}[Normalized bit reversing function for primes~\cite{halton1960efficiency}]\label{definition: Normalized Bit reversing function for prime numbers}
    Let $z$ be a positive integer and let~$p$ be a prime. 
    Consider the $p$-ary representation of~$z$, i.e., $z=\sum_{i\geq 0}a_ip^i$ with $a_i\in \llbracket p-1\rrbracket$. Then, the bit reversing function~$\phi_{p}$ is defined as $\phi_{p}(z)=\sum_{i\geq 0}a_ip^{-i-1}$. 
\end{definition}
Notice that unlike the bit-reversing function in the 2D case (Definition~\ref{definition:Bit reversing function}), the output of~$\phi_p$ is not necessarily an integer, and its domain is all integers.
\begin{definition}[Normalized HH set~\cite{halton1960efficiency}]\label{definition: Normalized Halton-Hammersely set}
    For a positive integer~$w$, define the normalized HH set $\bar{H}_w=\{(\frac{k}{w}, \phi_2(k), \phi_3(k))\}_{k=0}^{w-1}\subseteq [0,1]^3$.
    %be the normalized Halton-Hammersely (HH)
    %set of $w$ points whose coordinates are $(\frac{k}{w}, \phi_2(k), \phi_3(k))$ with $0\leq k\leq w-1$.
\end{definition}
%Notice that $\bar{H}_W\subseteq [0,1]^3$.
\begin{lemma}[Largest empty axis-parallel box in normalized HH-set]\label{lemma: Largest empty axis-parallel box in normalized Halton-Hammersely set}~\cite[Thm.~3]{Dumitrescu_Jiang_2012}
    The volume of the largest empty axis-parallel box in~$[0,1]^3$ that does not contain a point from~$\bar{H}_w$ is at most $\frac{24}{w}$.
\end{lemma}
%With Lemma~\ref{lemma: Largest empty axis-parallel box in normalized Halton-Hammersely set}, the volume for the largest empty axis-parallel box is bounded.
%Thus, when 
Lemma~\ref{lemma: Largest empty axis-parallel box in normalized Halton-Hammersely set} readily implies that after scaling up each coordinate to an integer, the volume of the largest empty axis-parallel box remains bounded.
\begin{definition}[Scaled HH set]\label{definition: Scaled  Halton-Hammersely set}
    For~$w=2^c$, define the scaled HH set inside the box~$[0,w]\times[0,w]\times[0,3^{\lceil c\log_3(2) \rceil}]$ as $H_w\triangleq\{(k, 2^c\phi_2(k), 3^{\lceil c\log_3(2) \rceil}\phi_3(k))\}_{k=0}^{w-1}$. 
    %\red{[The indices here start from~$0$, and in Definition~\ref{definition: Normalized Halton-Hammersely set} they start from~$1$. Decide on either one of them for consistency.]}\purple{fixed}%\mid 0\leq k\leq w-1\}$. 
    %Let $H_n$ be the Halton-Hammersely set of $n$ points with $n=2^c$ for some constant $c$. The scaled Halton-Hammersely set of points $\bar{H}_n$ have coordinates $(k, 2^c\phi_2(k), 3^{\lceil c\log_3(2) \rceil}\phi_3(k))$ with $0\leq k\leq n-1$
\end{definition}
%The scaling process can be considered as multiply $(2^c,2^c,3^{\lceil c\log_3(2) \rceil})$ to the corresponding side of $[0,1]^3$ which implies the volume for the largest empty axis-parallel box is bounded.
Lemma~\ref{lemma: Largest empty axis-parallel box in normalized Halton-Hammersely set} naturally extends to its scaled counterpart as follows.
\begin{lemma}\label{Lemma: Empty box volume in scaled Halton-Hammersely set}
    %For each scaled Halton-Hammersely-set of points $H_w$ with $w=2^c$ for some integer $c$, t
    The volume of the largest empty axis-parallel box in~$[0,w]\times[0,w]\times[0,3^{\lceil c\log_3(2) \rceil}]$ that does not contain a point from~$H_w$ is at most~$24w3^{\lceil c\log_3(2) \rceil}$.
\end{lemma}
After the scaling process, each coordinate from the scaled HH set is integer, which leads to the following definition.
\begin{definition}[Shifted Scaled HH set]\label{definition: shifted Scaled Halton-Hammersely set}
    %The $\bar{H}_n$ be the Halton-Hammersely set of $n$ points with $n=2^c$ for some constant $c$. 
    For an integer~$w=2^c$ and integers $d_1,d_2$ such that $0\leq d_1\leq w-1$ and $0\leq d_2\leq 3^{\lceil c\log_3(2) \rceil}-1$, define the shifted and scaled HH-set 
    \begin{align*}
        H_w^{d_1,d_2}\triangleq \{&(k, (2^c\phi_2(k)+d_1)\bmod w,\\
        &( 3^{\lceil c\log_3(2) \rceil}\phi_3(k)+d_2)\bmod 3^{\lceil c\log_3(2) \rceil}) \}_{k=0}^{w-1}.
        %\mid 0\leq k \leq w-1\}.
    \end{align*}    
    %and $0\leq d_2\leq 3^{\lceil c\log_3(2) \rceil}-1$
\end{definition}

Next, it is shown that the union of all shifted and scaled HH-sets covers all~$w^2m$ integer points of $\llbracket w-1\rrbracket\times \llbracket w-1 \rrbracket\times \llbracket 3^{\lceil c\log_3(2) \rceil}\rrbracket$, where~$m=3^{\lceil c\log_3(2) \rceil}$.
Clearly, each shifted and scaled HH-set contains~$w$ points, and there exists~$wm$ different shifted and scaled HH-sets, and hence it suffices to prove that any two distinct shifted and scaled HH-sets intersect trivially.

\begin{lemma}\label{lemma:shifted scaled hh set cover the whole space}
    If~$H_w^{d_1,d_2}\cap H_w^{d_1',d_2'}\ne\varnothing$ then~$(d_1,d_2)=(d_1',d_2')$.
\end{lemma}
\begin{proof}
    Since the intersection is not empty, it follows that there exist~$k,k'\in\zrange{w-1}$ such that
    \begin{align*}
        &(k, (2^c\phi_2(k)+d_1)\bmod w,( 3^{\lceil c\log_3(2) \rceil}\phi_3(k)+d_2)\bmod 3^{\lceil c\log_3(2) \rceil})=\\
        &(k', (2^c\phi_2(k')+d_1')\bmod w,( 3^{\lceil c\log_3(2) \rceil}\phi_3(k')+d_2')\bmod 3^{\lceil c\log_3(2) \rceil}).
    \end{align*}    
    %Suppose one point from $H_w^{d_1,d_2}$ coincide with a point from $H_w^{d_1',d_2'}$, i.e. $\{(k, (2^c\phi_2(k)+d_1)\bmod w,( 3^{\lceil c\log_3(2) \rceil}\phi_3(k)+d_2)\bmod 3^{\lceil c\log_3(2) \rceil})=\{(k, (2^c\phi_2(k)+d_1')\bmod w,( 3^{\lceil c\log_3(2) \rceil}\phi_3(k)+d_2')\bmod 3^{\lceil c\log_3(2) \rceil})\}$ for some $0\leq k\leq w-1 $.
    Thus,~$k=k'$, and 
    \begin{align*}
        (2^c\phi_2(k)+d_1)\bmod w=(2^c\phi_2(k)+d_1')\bmod w,
    \end{align*}
    which implies that $(d_1-d_1')\bmod w=0$.
    Since $0\le d_1,d_1'\le w-1$, it follows that $d_1=d_1'$.
    The equality between $d_2$ and~$d_2'$ is proved similarly.
\end{proof}
%Notice that $\cup_{0\leq d_1\leq n-1,0\leq d_2\leq }$ covers $n^3$ points of the box $[0,n)\times [0,n)\times [0,3^{\lceil c\log_3(2) \rceil}))$. Further, it is easy to prove that the volume for the largest empty axis-parallel box is still bounded.
\begin{lemma}\label{Lemma: Empty box volume in shifted scaled Halton-Hammersely set}
    For each shifted scaled HH set $H_w^{d_1,d_2}$ with $w=2^c$ for some integer $c$, the volume of the largest empty axis-parallel box in $[0,n]\times[0,n]\times[0,3^{\lceil c\log_3(2) \rceil}]$ that does not contain a point from $H_w^{d_1,d_2}$ is at most $24w3^{\lceil c\log_3(2) \rceil}$.
\end{lemma}
\begin{proof}
    Similar to the proof to Lemma~\ref{lemma:shifted vdc area}, the details are omitted for brevity.
\end{proof}
Previous lemmas bound the volume of the largest empty box~$B$ that does not contain a point from the HH set or its shifts.
%the shift of HH sets.
In what follows, we also require tiling of an HH set in a larger codeword similar to Lemma~\ref{union of vdc lemma}.
\begin{lemma}\label{lemma:union of HH set}
        Let~$w=2^c$, let~$v_1$ be an integer multiple of~$w$, and let~$v_2$ be an integer multiple of $3^{\lceil c\log_3(2) \rceil}$. 
        For~$x,y\in\llbracket v_1/w-1\rrbracket$ and~$z\in\llbracket v_2/3^{\lceil c\log_3(2) \rceil}-1\rrbracket$, let~$H_w^{(x,y,z)}\triangleq \{(z_1+xw,z_2+yw,z_3+z3^{\lceil c\log_3(2) \rceil})\vert (z_1,z_2,z_3)\in H_w\}$, and let~$T\triangleq \bigcup_{x,y\in \llbracket v_1/w-1\rrbracket, z\in \llbracket v_2/3^{\lceil c\log_3(2) \rceil}-1\rrbracket}H_{w}^{(x,y,z)}$. 
    Then the volume of the largest empty box that does not contain a point from~$T$ is at most~$24w3^{\lceil c\log_3(2) \rceil}$.
\end{lemma}
\begin{proof}
    Similar to the proof of Lemma~\ref{union of vdc lemma}, the details are omitted for brevity.
\end{proof}
Lemma~\ref{lemma:union of HH set} shows that tiling small HH sets in a large cuboid does not significantly affect the area of largest empty box. 
We will also require shifts of~$T$, i.e., for $d_1\in \zrange{w-1}$ and $d_2\in \zrange{3^{\lceil c\log_3(2) \rceil}-1}$, let 
$$T_{d_1,d_2}\triangleq\{(\alpha,(\beta+d_1)\bmod v_1,(\gamma+d_2)\bmod v_2)\vert (\alpha,\beta,\gamma)\in T \}.$$
It follows from Lemma~\ref{lemma:shifted scaled hh set cover the whole space} that $\bigcup_{d_1\in \zrange{w-1},d_2\in \zrange{3^{\lceil c\log_3(2) \rceil}-1}} T_{d_1,d_2}$ covers all entries in the $v_1\times v_1\times v_2$ cuboid. In addition, it follows from previous lemmas that the volume of the largest empty axis-parallel box in $[0,v_1]\times[0,v_1]\times[0,v_2]$ that does not contain a point from $T_{d_1,d_2}$ is at most $24w3^{\lceil c\log_3(2) \rceil}$

\subsection{Encoding}\label{section:encoding for 3D}
%Encodings in Section~\ref{section:encoding for 2D} can be applied to 3D.
The encoding in 3D operates according principles similar to Section~\ref{section:encoding for 2D}.
Before introducing the zero square identification encoding in 3D, we first illustrate the coloring function. 
Recall that the codeword $C\in \Sigma^{n\times n\times n'}$ is partitioned into smaller $d\times d\times d$ subcubes called units.
A coloring is then given to different units such that units with the same color contain the same information.
We color the $\frac{n}{d}\times \frac{n}{d}\times\frac{n'}{d}$ grid of units in $\frac{ab}{24}$ colors (where~$a$ and~$b$ are as in Section~\ref{section: Problem definition for 3D}) by first coloring a smaller $a/8\times a/8\times b/3$ grid of units and then tiling the whole $\frac{n}{d}\times \frac{n}{d}\times\frac{n'}{d}$ grid.

By the definition of shifted scaled HH-set (Definition~\ref{definition: shifted Scaled Halton-Hammersely set}) and by Lemma~\ref{lemma:shifted scaled hh set cover the whole space}, the union of all $H_w^{d_1,d_2}$ with $d_1\in \zrange{a/8-1}$ and $d_2\in \zrange{b/3-1}$ covers all $\frac{a^2b}{64\cdot3}$ integer points of $\zrange{a/8-1}\times\zrange{a/8-1}\times\zrange{b/3-1}$. 
Then, we color all integer points in a set $H_w^{d_1,d_2}$ with color $\frac{d_1b}{3}+d_2$. 
This coloring covers all integer points in $\zrange{a/8-1}\times\zrange{a/8-1}\times\zrange{b/3-1}$ and is injective due to Lemma~\ref{lemma:shifted scaled hh set cover the whole space}.
Since the small $a/8\times a/8\times b/3$ grid of units is colored and since $n|da$, $n'|db$, we can use the small grid to tile the whole $\frac{n}{d}\times \frac{n}{d}\times\frac{n'}{d}$ grid, where the coloring remains identical.
After tiling, $T_{d_1,d_2}$ with $d_1\in \zrange{a/8-1}$ and $d_2\in \zrange{b/3-1}$ is colored by $\frac{d_1b}{3}+d_2$ according to the definition of $T$ and $T_{d_1,d_2}$ in Lemma~\ref{lemma:union of HH set}.
%using a coloring function $\chi:\zrange{\frac{n}{d}}\times \zrange{\frac{n}{d}}\times\zrange{\frac{n'}{d}}\xrightarrow[]{}\alpha$ defined as
% $$\chi(i,j,k)=$$

In zero square identification encoding, we place a $d\times d\times d$ zero cube in all units of color~$0$.
In all other units, we fix certain $8$ bits to be $1$ so that the only $d\times d\times d$ zero cube in the entire codeword are the units with color $0$.
In detail, let~$\boldx\in\Sigma_q^k$ be the message to be encoded, where $k=(\frac{ab}{24}-1)(R-\log_q(\frac{ab}{24}-1))$ and $R=d^3-8$.
We begin by mapping~$\boldx$ to~$\frac{ab}{24}-1$ \textit{distinct} strings~$\boldx_i$'s, each containing~$R$ bits, using an invertible function~$g_1(\boldx)=\{\boldx_i\}_{i=0}^{\frac{ab}{24}-1}$ which simply appends indices to the parts of~$\boldx$, as done in Proposition~\ref{proposition:g1}.
% \begin{lemma}\label{lemmal:g1 for 3D}
%     The function $g_1$ can be computed in polynomial time.
% \end{lemma}
% \begin{proof}
%     Identical to the proof for Lemma~\ref{lemma:g_1'}.
% \end{proof}
Next, we introduce a function $g_2:\Sigma_q^R\xrightarrow{}\Sigma_q^{d\times d\times d}$, which sets all the corners of the unit cube to be $1$, and the remaining $d^3-8$ bits of the output are the $R$ bits of the input.
We then let~$\boldX_i\triangleq g_2(\boldx_i)\in\Sigma_q^{d\times d\times d}$ for each~$i\in\{1,\ldots,\frac{ab}{24}-1\}$, and embed the~$\boldX_i$ into the $n\times n\times n'$ codeword according to the coloring described above.
\subsection{Decoding}\label{section: decoding for 3D}
Given a legal fragment from~$C(\boldx)$ to decode, the decoder first identifies all the units contained in the fragment by aligning them against~$C(\boldx)$.
Since~$C(\boldx)$ is unknown, the decoder identifies the boundaries between all units 
%as follows.
%The decoder identifies the boundaries 
by finding a~$d\times d\times d$ zero cube in the fragment, i.e., the decoder identifies the units with color~$0$. 
It remains to prove there exist a unit with color $0$ in any legal fragment, and that any~$d\times d\times d$ zero cube must be a unit, which is done in Lemma~\ref{lemma:enough units in 3D} and Lemma~\ref{lemma:distinct units in 3D} which follow.
Having identified the boundaries, the decoder puts all units in a set~$\{\boldY_i\}_{i=1}^s$ of some size~$s$ (without repetitions). 
The decoder applies~$g_2^{-1}$ on each~$\boldY_i$, and later applies~$g_1^{-1}$  on the resulting set. 
That is, the output of the decoding algorithm is~$\hat{\boldx}=g_1^{-1}(\{g_2^{-1}(\boldY_i)\}_{i=1}^s)$.
We now turn to prove that the decoding algorithm is well-defined, that the alignment succeeds, and that~$\hat{\boldx}=\boldx$.

To prove these statements, we show that the fragment contains all~$\frac{ab}{24}$ distinct units.
This is done in two steps: we first show that the fragment contains sufficiently many units (Lemma~\ref{lemma:enough units in 3D}) and then show that all~$\frac{ab}{24}$ \textit{distinct} units are among them (Lemma~\ref{lemma:distinct units in 3D}).
\begin{lemma}\label{lemma:enough units in 3D}
    A legal fragment contains an $(x+1)\times (y+1)\times(z+1)$ grid of units for some positive integers~$x,y,z$ such that $xyz=ab$.
\end{lemma}

\begin{proof}
We argue by contrapositive.  
Fix a legal fragment and positive integers \(x,y,z\) such that the
fragment contains an \((x+1)\times(y+1)\times(z+1)\) unit grid.
Write
\[
V(x,y,z):=(xd+2d-1)(yd+2d-1)(zd+2d-1)
\]
for the fragment’s volume implied by these parameters.  We will show
that if \(xyz<ab\) then
\[
V(x,y,z)<(3d-1)^2(abd+2d-1),\tag{1}\label{eq:target}
\]
contradicting the minimality condition in the definition of a legal
fragment
For any triple of positive integers with \(xyz<ab\) we have
\[
x+y+z \;\le\; 1+1+xyz \;<\; ab+2, \qquad
xy+yz+zx \;\le\; 1+2xyz \;<\; 2ab+1.
\]
A direct expansion yields
\begin{align*}
V(x,y,z)
&=(xd+2d-1)(yd+2d-1)(zd+2d-1)\\
&=xyz\,d^{3}+(xy+yz+zx)d^{2}(2d-1)\\
&\quad +(x+y+z)d(2d-1)^{2}+(2d-1)^{3}.
\end{align*}

Substituting the inequalities gives
\begin{align*}
V(x,y,z)
&<ab\,d^{3}+(2ab+1)d^{2}(2d-1)\\
&\quad +(ab+2)d(2d-1)^{2}+(2d-1)^{3}\\
&=(3d-1)^2(abd+2d-1).\tag*{\qedhere}
\end{align*}
\end{proof}

\begin{lemma}\label{lemma:distinct units in 3D}
    In the above process, we have $s=\frac{ab}{24}$.
\end{lemma}
\begin{proof}
    We show that the cuboid whose existence is guaranteed by Lemma~\ref{lemma:enough units in 3D} contains all~$\frac{ab}{24}$ distinct  units.
    Let~$P$ be the $(x+1)\times(y+1)\times(z+1)$ grid of  units guaranteed by Lemma~\ref{lemma:enough units in 3D}. Define a $(x+1)\times(y+1)\times(z+1)$ grid by contracting each unit to a point.
    With this contracting process, the volume of~$P$ is~$xyz=ab$.

    By Lemma~\ref{lemma:union of HH set}, for every $d_1\in \zrange{a/8}$ and every $d_2\in \zrange{b/3}$, the set~$T_{d_1,d_2}$ is colored by $\frac{d_1b}{3}+d_2$.
    %according to the definition of $T$ and $T_{d_1,d_2}$ in .
    Suppose there exist some color $\beta$ that is not in the cuboid from Lemma~\ref{lemma:enough units in 3D}.
    Lemma~\ref{lemma:union of HH set} implies that the volume of any axis-parallel box that does not contain the color $\beta$ is strictly less than $24\cdot\frac{ab}{24}$.
    Since the volume of $P$ is~$ab$, it follows that it contains the color~$\beta$, a contradiction. 
\end{proof}
\subsection{Bounds and Rates}
\begin{theorem}\label{theorem:rate for 3D} 
In a parameter regime where $M\rightarrow \infty$, $h=o(\sqrt[3]{M})$, and $h = \omega(\sqrt[3]{\log_q M})$,
    the rate of the code in Section~\ref{section:encoding for 3D} approaches~$\frac{1}{1296}$. 
\end{theorem}
\begin{proof}
    Recall that for any~$N=n^2\cdot n'$, any~$M\le N$, and any~$h$, we let~$d$ be the largest integer such that~$h\ge 3d-1$, and let~$c$ be the largest integer such that $a=2^c$ and $b=3^{\lceil c\log_3(2)\rceil}$ satisfy $M\ge (3d-1)^2((ab+1)d+d-1)$. 
    Therefore, 
    \begin{equation}\label{equation: bound for M in 3D}
        (3d-1)^2((ab+1)d+d-1)\le M\le (3d-1)^2((6ab+1)d+d-1)
    \end{equation}
    and hence 
    \begin{equation}\label{equation: bound for ab in 3D}
        \frac{1}{6d}\left(\frac{M}{(3d-1)^2}+1\right)-1\le ab \le \frac{1}{d}\left(\frac{M}{(3d-1)^2}+1\right)-1.
    \end{equation}
    
    Since~$d=\Theta(h)$, and since $h=o(\sqrt[3]{M})$, it follows that
    $$\lim_{M\rightarrow\infty}\frac{1}{6d}\left(\frac{M}{(3d-1)^2}+1\right)= \infty,$$ and thus in this parameter regime we have that $ab\xrightarrow[]{}\infty$.    
    Further, since~$d=\Theta(h)$, since $h = \omega(\sqrt[3]{\log_q ab})$, and since~$M\ge ab$, it follows that $$0\le \lim_{M\rightarrow\infty}\frac{\log_q(ab)}{d^3}\le \lim_{M\to \infty}\frac{\log_q(M)}{d^3}= 0,$$ 
    therefore $\lim_{M\rightarrow\infty}\frac{\log_q(ab)}{d^2}=0$.

     Now, recall that encoding requires input of length~$k=(\frac{ab}{24}-1)(d^3-8-\log_q(\frac{ab}{24}-1))$.
    Therefore, the resulting rate is
    \begin{equation}\label{equation:rate calc for 3D}
        \frac{((\frac{ab}{24}-1)(d^3-8-\log_q(\frac{ab}{24}-1))}{M}.
    \end{equation}
    With~\eqref{equation: bound for M in 3D}, the rate is at least 
    $$\frac{((\frac{ab}{24}-1)(d^3-8-\log_q(\frac{ab}{24}-1))}{((3d-1)^2((6ab+1)d+d-1)}$$ which is equal to
    $$\frac{
1 - \frac{24}{ab} - \frac{8}{d^3} - \frac{\log_q\left(\frac{ab}{24} - 1\right)}{d^3} + \frac{192}{ab d^3} + \frac{24 \log_q\left(\frac{ab}{24} - 1\right)}{ab d^3}
}{
1296 + \frac{432}{ab} - \frac{864}{d} - \frac{504}{ab d} + \frac{144}{d^2} + \frac{192}{ab d^2} - \frac{24}{ab d^3}
}
.$$
Thus, the rate is at least $$\frac{
1 - \frac{24}{ab} - \frac{8}{d^3} - \frac{\log_q\left(\frac{ab}{24} - 1\right)}{d^3} 
}{
1296 + \frac{432}{ab} + \frac{144}{d^2} + \frac{192}{ab d^2} 
}
$$ and since $ab,d\overset{M\to\infty}{\to}\infty$ and $\log_q(ab)/d^3\xrightarrow[]{}0$, it follows that the rate approaches $\frac{1}{1296}$.
\end{proof}

A proof mirroring that of Theorem~\ref{thm: random encoding} shows that  there exist codes of rates approaching~$1$. 
For brevity, we omit the details. 
\begin{remark}
   When $M,h$ are given such that we can find $m,d,a,b$ with $M=(3d-1)^2((ab+1)d+d-1)$, 
   %the rate of the~$(q,N,M,h)$-forensic code in Section~\ref{section:encoding for 3D} approaches~$\frac{1}{216}$, 
   a proof similar to that of Theorem~\ref{theorem:rate for 3D} shows that the rate of the~$(q,N,M,h)$-forensic code in Section~\ref{section:encoding for 3D} approaches~$\frac{1}{216}$; this follows since~$M$ would be equal to its lower bound in~\eqref{equation: bound for M in 3D} and since~$ab$ would be equal to its upper bound in~\eqref{equation: bound for ab in 3D}, and hence a more accurate bound on the rate would be achieved in~\eqref{equation:rate calc for 3D}.
   Table~\ref{table:3D} presents some examples.
\end{remark}
\begin{table}[h]
\caption{Example values of $M,h,a,b$, and the resulting code rate (converging toward $1/216$).}\label{table:3D}
\centering
\begin{tabular}{r r r r r}
\toprule
$M$ & $h$ & $a$ & $b$ & \textbf{rate} \\
\midrule
209{,}935      & 11  & 16  & 27  & 0.004203745 \\
425{,}124      & 14  & 16  & 27  & 0.004515184 \\
752{,}267      & 17  & 16  & 27  & 0.004608089 \\
1{,}836{,}159  & 23  & 16  & 27  & 0.004628419 \\
2{,}639{,}780  & 26  & 16  & 27  & 0.004616867 \\
5{,}082{,}084  & 14  & 64  & 81  & 0.004621950 \\
\bottomrule
\end{tabular}

\end{table}

\section{Error-resilient forensic coding for two dimensional information}\label{section:error correction}
\subsection{Problem Definition}
With a $(q, n, M, h)$-forensic code, the decoder can reconstruct the message given any legal fragment.
Nevertheless, various noise patterns may occur in practical settings; for instance, some bits may inadvertently flip from~$0$ to~$1$ or vice versa, a phenomenon we refer to as \textit{bit-flips}.
Bit-flips might emerge in 3D-printing technology as a result of imprecise printing hardware, or noisy reading mechanisms (see~\cite{wang2024secureinformationembeddingextraction}).
Moreover, depending on the particular method of bit embedding, the adversary might intentionally induce bit-flips in the fragment. 
Therefore, we aim for worst-case guarantees in the form of an additional parameter~$\delta$, which quantifies the maximum number of bit-flips allowed in the entire codeword.
%—a substitution error where the bit’s value is altered.
%In addition, noise may cause erasures, where bits are lost or replaced by an indeterminate symbol. 
%Previous literature study the case that an adversary might intentionally induce such bit-flips or erasures, but focused on codeword being a one-dimensional string, like the bit-flips in sliced channel \cite{sima2023robustindexingslicedchannel} and optimal deletion correcting codes~\cite{sima2020optimal}. 
%Our paper extends the above line of works and provides a new scheme for error correction codes of higher dimension data. 

%The technology we consider is based on~$(2, n, M, h)$-forensic codes (binary instead of $q$-ary), which are designed to reconstruct the original message from any legal fragment of the transmitted or stored codeword. 
%We focus on worst-case guarantees against bit-flips, with bit-flips quantified by a parameter~$\delta$ representing the maximum number of substitutions allowed in the entire codeword.
Clearly, in the worst-case all~$\delta$ flips are concentrated in the fragment received by the decoder. 
Our goal is constructing encoding and decoding functions which enable the decoder to retrieve the message in all cases where the received fragment contains a rectangle of area at least~$M$ with sides of length at least~$h$, where the adversary can flip any~$\delta$ bits in the entire codeword.
We call such codes $(q, n, M, h, \delta)$-forensic; for simplicity we focus on the case~$q=2$, though a generalization to all~$q$ is possible.

Handling~$\delta$ bit-flips is challenging since these flips might interfere with the synchronization process.
To address this challenge we employ recent advances in coding for the recently defined \textit{sliced channel model}~\cite{sima2023robustindexingslicedchannel} (described next), as well as a specialized form of the synchronization process discussed earlier.

%even if the codeword is broken into fragments.
% \blue{[Write a short paragraph describing the technology and the kind of noise patterns that might emerge, also mention that in some implementations (existing or future) the adversary might be able to induce bit-flips or erasures (and therefore we focus on worst-case guarantees). Say that we focus on but flips and erasures. Give definitions.]}
% \red{Nevertheless, noises may be introduced in this process.
% That is, we  model the perpetrator receives $C\in\Sigma_q^{n\times n}$, breaks it apart, and discards all fragments except one. 
% That one fragment arrives at the decoder may encounter some noises, i.e. substitution errors like bit-flipping.}

%We quantify the adversary's ability to introduce substitutions by an additional parameter~$\delta$, i.e., the adversary can introduce at most~$\delta$ but flips in the entire codeword. 

%Given parameter $\delta$ representing the maximum number of substitution errors that a valid code can correct, we may assume all substitution errors happens on a single fragment that the decoder receives. 
%With the same setting on $q, n, M, h$ in \ref{section:overview}, w

%introduce at most~$\delta$ bit-flips anyafter correcting at most $\delta$ adversary errors.

%as we not only need to recover erroneous information, we also need to reconstruct if an error is introduced in the synchronization process. 

%use error correction codes like Reed-Solomon codes and
%robust indexing technique by~\cite{sima2023robustindexingslicedchannel}, which provides an optimal solution for the sliced channel problem, 

\subsection{The Sliced Channel}
%As mentioned in previous sections, the robust indexing method plays a key role in our scheme. 
%The robust indexing studies the problem that the data is given as a binary string and encoded in an unordered set of strings, that is, $M$ strings of length $L$. The decoder wishes to recover data from erroneous version of strings with at most $K$ errors, which contain substitutions.
Motivated by applications in DNA storage, the \textit{sliced channel model} has been studied in multiple papers~\cite{sima2021coding,sima2023robustindexingslicedchannel,wang2024breakresilientcodesforensic3d,tornpaper}.
For parameters~$Q$ and~$L$, in this channel
%Robust indexing addresses the problem where 
data is represented as a binary string and encoded into an unordered set of~$Q$ strings, each of length~$L$.
The decoder aims to reconstruct the original data from erroneous versions of these~$Q$ strings, which may be corrupted by up to~$K$ substitutions overall, where~$K$ is an additional parameter.
Specifically, the following theorem presents an explicit code construction for the sliced channel, capable of correcting any~$K$ substitutions with asymptotically optimal redundancy.
%, using a technique called robust indexing.

\begin{theorem}\cite{sima2023robustindexingslicedchannel}\label{Theorem: cited sliced channel}
For integers $Q$, $L$, $K$, and $L' \triangleq 3 \log Q + 4K^2 + 2$, if $L' + 4KL' + 2K \log(4KL') \leq L$, there exists an explicit binary $K$-substitution correcting code for the sliced channel, computable in $\operatorname{poly}(Q, L, K)$ time, that has  
$$ 2K \log QL + (12K + 2) \log Q + O(K^3) + O(K \log \log QL) $$
redundant bits.
\end{theorem}
In what follows we employ Theorem~\ref{Theorem: cited sliced channel} to introduce a $(2, n, M, h,\delta)$-forensic code with $\delta=ch^{2/3}$ and $M=ch^22^{h^{4/3}}$ where $c$ is some small constant stemming from Lemma~\ref{lemma:parametrsAreGood} which follows.

There are two challenges associated with constructing bit-flip resilient forensic codes, depending on the locations of substitution errors: (I) bit-flips introduced into the alignment bits (i.e., the ``dead'' bits that facilitate the identification of complete units), and (II) bit-flips introduced into the information bits.
%In the design of error correction encoding, concepts from the zero square identification encoding (Section~\ref{section:encoding for 2D}) are modified to address these challenges.
To address (I) we add more alignment bits until the Hamming distance between the erroneous and correct versions of units is large enough to allow bit-flip correction.
To address (II) we use Theorem~\ref{Theorem: cited sliced channel} to encode the message bits. 
%We discuss the these solutions with more details in the following section. 

\subsection{Encoding}\label{section:ECencoding}

%\red{[A couple of sentences reminding the reader that we partition to~$d\times d$ rectangles, color them, etc. Also remind about parameters.]}
%Recall the definitions from the Section~\ref{section:overview} and similar encoding schemes from Section~\ref{section:encoding for 2D}, we begin by partitioning the $n\times n$ codeword into smaller $d\times d$ units, where \( d \) is the largest integer such that \( h \geq 3d-1 \). 
Recall that in zero square identification encoding (Section~\ref{section:encoding for 2D}), we begin by partitioning the $n\times n$ codeword into smaller $d\times d$ units, where \( d \) is the largest integer such that \( h \geq 3d-1 \), and let \( m \) be the largest integer power of~$2$ such that \( M \geq (4d-1)((m+1)d+d-1) \);
these parameters are chosen to ensure that each fragment will contain at least $m$ units.
Then, the~$\frac{n}{d}\times \frac{n}{d}$ grid of units is colored in~$m'\triangleq\frac{m}{4}$ colors according to the coloring function in~\eqref{coloring function}.

% \blue{To correct~$\delta$ substitutions, we encode indexing bits as follows.
% We place $d\times d$ zero squares in all units of color~$0$ to serve as a synchronization marker, similar to the zero square identification encoding in~Section~\ref{section:encoding for 2D}.
% %---where a unit of color~$0$ serves as a synchronization marker. 
% In all other units, we fix all~$4d-4$ bits on the sides of the $d\times d$ unit as~$1$. }
%certain~$4d-4$ bits to be~$1$, that is, all bits on the side of $d\times d$ square become~$1$. 
%This process helps the decoder to resist errors in the indexing bit while an detailed proof is given in Section~\ref{Section: error correction decoding}.   

%To encode the information bits, let $\boldx\in \Sigma_q^k$ be the message to be encoded,
To construct a $(2, n, M, h, \delta)$-forensic code with $\delta=ch^{2/3}$ and $M=h^22^{ch^{4/3}}$ where $c$ is some small enough constant, we apply Theorem~\ref{Theorem: cited sliced channel} with the parameters
\begin{align}\label{equation:parameters}
    Q&=m'-1,\nonumber\\
    L&= d^2-(4d-4), \mbox{ and}\nonumber\\
    K&=\delta.
\end{align}
%Let $\delta$ be the number of substitution errors that this encoding can resist. 
%Recall $m'\triangleq\frac{m}{16}$ from ~\ref{section:encoding for 2D}
%To utilized Theorem~\ref{Theorem: cited sliced channel}, we define
%and let~$Q=m'-1$ and $L= d^2-(4d-4)$.
%, and $R=2\delta \log \hat{m}L + (12\delta + 2) \log \hat{m} + O(\delta^3) + O(\delta \log \log \hat{m}L)$.

\begin{lemma}\label{lemma:parametrsAreGood}
The parameters in~\eqref{equation:parameters} satisfy the requirements in Theorem~\ref{Theorem: cited sliced channel}.    
\end{lemma}

\begin{proof}
See Appendix~\ref{section:omittedProofs}, Lemma~\ref{lemma:parametrsAreGoodappendix}.
\end{proof}
Let~$k$ be the base-$2$ logarithm of the size of the code implied by Lemma~\ref{lemma:parametrsAreGood}, and let~$\operatorname{Enc}:\{0,1\}^k\to \binom{\{0,1\}^L}{M}$ and~$\operatorname{Dec}:\binom{\{0,1\}^L}{M}\to\{0,1\}^k$ be the respective encoding and decoding functions.

We begin the encoding mechanism for the~$(2,n,M,h,\delta)$ forensic code by sending the input information$\boldx$ to $\operatorname{Enc}$ and receiving an unordered set~$\{\boldx_1,\ldots,\boldx_M\}$ of~$M$ strings of length~$L$ each. 
Each string is then mapped a~$d\times d$ unit using a function $g_2$ as follows. 
%The function $g_2:\Sigma_2^L\xrightarrow[]{}\Sigma_2^{d\times d}$ sets all sides of the~$d\times d$ output as~$1$, and the remaining $d^2-(4d-4)$ bits of the output contain the~$L$ bits of the input at some fixed arbitrary order. 
The function $g_2:\Sigma_2^L\xrightarrow[]{}\Sigma_2^{d\times d}$ sets all boundary bits of the output (i.e., top row, bottom row, leftmost column, and rightmost column) to~$1$, and the remaining $d^2-(4d-4)$ bits of the output contain the~$L$ bits of the input at some fixed arbitrary order.
We then let~$\boldX_i\triangleq g_2(\boldx_i)\in\Sigma_2^{d\times d}$ for each~$i\in\{1,\ldots,m'-1\}$, and embed~$\boldX_i$ into the $n\times n$ codeword according to a process identical to the one in Section~\ref{section:encoding for 2D}, where~$\boldX_i$ is embedded in all the~$\frac{n}{d}\times \frac{n}{d}$ grid points colored by~$i$, and 
%as the same coloring function from Section~\ref{section:encoding for 2D}. 
% Recall the coloring function:
% \begin{align*}
%     c(i,j)=(j-f_{m'}(i \bmod m'))\bmod m'
% \end{align*}
% for every~$(i,j)\in \llbracket n/d\rrbracket\times \llbracket n/d\rrbracket$, where~$f_{m'}$ is the bit-reversal function mentioned in Definition~\ref{definition:Bit reversing function}. 
% Then, we place~$\boldX_i'$ in all units colored by~$i$ where $\{\boldX_i\}_{i=1}^{m'-1}$ covers colors $\{1, 2, \ldots, m'-1\}$ while 
all entries of all units of color~$0$ are set to~$0$.

\subsection{Decoding}\label{Section: error correction decoding}
Given a legal fragment taken from a codeword $C(\boldx)$ which was corrupted by at most~$\delta$ bit-flips, the decoder first identifies all the units contained in the fragment by aligning them against $C(\boldx)$. 
Specifically, the decoder identifies the boundaries by finding the lowest Hamming weight $d\times d$ square in the fragment.
%Since the existence of a unit with color $0$ is proved in Section~\ref{section:decoding for 2D}, we can always find the unit with color~$0$.
It remains to prove that this lowest Hamming weight square aligns with a~$0$-colored unit in~$C(\boldx)$.
%i.e., 
%any unit colored in color other than~$0$, is different from the unit colored $0$, i.e. the respective %$d\times d$ squares are not identical due to bit-flips.
%for units colored $0$ and units colored other than 0. (Otherwise color~$0$ square is not unique.)
This is done in Lemma~\ref{lemma:error correcting uniquely identify color 0} below.
Having identified the boundaries for a $0$-colored unit, the decoder identifies the boundaries of all other units.
Since the fragment contains units of all colors, as proved in Section~\ref{section:decoding for 2D}, at least one (potentially corrupted) unit~$Y_i$ of color~$i$ can be found for all $i\in \{1,2,\ldots, m'-1\}$, and then the sliced-channel decoder is used to correct bit-flips.

Specifically, the decoder applies~$g_2^{-1}$ on each~$\boldY_i$, and later applies~$\operatorname{Dec}$  on the resulting set. 
That is, the output of the decoding algorithm is~$\hat{\boldx}=\operatorname{Dec}(\{g_2^{-1}(\boldY_i)\}_{i=1}^{m'-1})$, where each~$\boldY_i$ is a unit of color~$i$.
We now turn to prove that the decoding algorithm is well-defined, that alignment succeeds, and that~$\hat{\boldx}=\boldx$.

\begin{lemma}\label{lemma:error correcting uniquely identify color 0}
    The lowest Hamming weight~$d\times d$ square in the fragment aligns with a~$0$-colored square in~$C(\boldx)$.
    %The $0$-colored unit can be uniquely identified.
\end{lemma}

\begin{proof}
    Recall that $\delta=ch^{2/3}$ for some small constant~$c$, and that~$d$ is the largest integer such that $h\geq 3d-1$, hence $d\geq\lfloor\frac{h}{3}\rfloor-1$ and $\delta\leq 3cd^{2/3}$,
    %Therefore, it follows that 
    i.e., at most~$3cd^{2/3}$ bit-flips occurred.
    The $0$-colored unit is a~$d\times d$ zero square, while units of all other colors have Hamming weight at least~$4d-4$ due to the $1$'s on their sides (see Section~\ref{section:ECencoding}).
    Hence, the minimum Hamming distance between a~$0$-colored unit and a unit of any other color is at least~$4d-4$.
    %Therefore, since~$3cd^{2/3}<\floor{\frac{4d-4}{2}}$, every pattern of at most~$3cd^{2/3}$ bit-flips cannot cause any confusion between a~$0$-colored unit and a non~$0$-colored unit. 
    %In order to the minimum Hamming-weight square is exactly a~$0$-colored unit, we also need to consider prove the minimum Hamming-weight square align with the boundaries of units.
    We now show that the minimum Hamming-weight square in the fragment aligns with a~$0$-colored unit.
    %, we also need to consider prove the minimum Hamming-weight square align with the boundaries of units.
    First, if the minimum Hamming-weight square aligns with the boundaries of some unit of color other than~$0$, we arrive at a contradiction to the above statement regarding the minimum Hamming distance between units.
    That is, since~$3cd^{2/3}<\lfloor \frac{4d-4}{2} \rfloor$ for sufficiently large~$d$, no pattern of at most~$3cd^{2/3}$ bit-flips can cause a~$0$-colored unit to have higher Hamming weight than a non-$0$ colored unit.
    Second, if the minimum Hamming-weight square does not align with the boundaries of any unit, it must intersect at least two units. As the $1$'s are aligned on the boundary of each unit, the minimum Hamming-weight square intersects with the boundaries of different units and thus contains at least $d$ entries with~$1$'s .
    Again, this leads to a contradiction to the minimum distance property of units, since for large enough~$d$, no pattern of at most~$3cd^{2/3}$ bit-flips can cause a~$0$-colored unit to have Hamming weight higher than~$\lfloor d/2 \rfloor$.   
    % This makes the Hamming weight of the minimum Hamming-weight square be at least $d$ which is greater than~$3cd^{2/3}$, contradiction.
    % This shows the hamming distance between them is at least $4d-4$ which is possible to resist $3cd^{2/3}$ bit-flips for sufficiently small~$c$.
\end{proof}
After the boundaries of units with color $0$ are identified, we identify the boundaries of units of all other colors.
%from the coloring function~\eqref{coloring function}. 
Thus, the function $g_2^{-1}:\Sigma_q^{d\times d}\rightarrow\Sigma_q^{(d-1)\times (d-1)}$ removes the boundaries from the input square, and vectorizes the result. 
The resulting vectors are given as input to the function~$\operatorname{Dec}$ from~\cite{sima2024robust}. 
Since there are at most~$\delta$ bit-flips in these vectors, Theorem~\ref{Theorem: cited sliced channel} guarantees correct decoding.

% and flattens c, and the resulting squares are given as input to 
% Given a unit $\boldY_i\in\Sigma_q^{d\times d}$, where $\Sigma_q^{d\times d}$ is indexed using $\llbracket d \rrbracket^2$, define $g_2^{-1}(\boldY_i)$ as follows. 
%     \begin{enumerate}
%         \item Extract row~$0$, row~$d-1$, column~$0$, and column~$d-1$.
%         \item Output the rest $(d-1)\times (d-1)$ square.
%     \end{enumerate}    
% Finally, applying~$\operatorname{Dec}$ as the decoding section from Page 5 in~\cite{sima2023robustindexingslicedchannel} which recovers the message.
 \subsection{Bounds and rates}\label{section:bounds and rates for ecc}

 %\red{[Need to say something about the rate we achieve.]}
In this section we provide necessary and sufficient bounds for a $(q, n, M, h,\delta)$-forensic code to exist. 
We begin with the observation that a $(q, n, M, h,\delta)$-forensic code needs to have large Hamming distance and then apply the sphere packing bound.

\begin{theorem}\label{thm: necessary bounds for error correcting}
    The rate of a $(q, n, M, h,\delta)$-forensic code cannot exceed $\frac{M-\log(\sum_{i=0}^{\delta}\binom{M}{i}(q-1)^i)}{M}$.
\end{theorem}
\begin{proof}
    Let $C$ be any $(q,n,M,h,\delta)$-forensic code which can correct $\delta$ bit-flips, and let~$\mu$ be an integer such that such that
    %As two fragments with same area $M$ shall not contain the same information after $\delta$ bit-flips, the Hamming distance between 
    $\lfloor\frac{\mu-1}{2}\rfloor=\delta$. % where $d$ is the minimum Hamming distance required to avoid confusable codewords.
    Let~$R$ be any rectangle inside~$\zrange{n}\times \zrange{n}$ of area~$M$ and side length at least~$h$ (e.g., the top-left one), and consider the code~$C_R$ given by restricting all codewords of~$C$ to~$R$; clearly~$|C_R|\le |C|$.

    If~$|C_R|<|C|$, it follows that there exist two codewords~$\boldX,\boldY\in C$ which coincide on~$R$, contradicting the forensic capabilities of~$R$ (i.e., $\boldX,\boldY$ restricted to~$R$ are equal legal fragments, and hence~$\boldX$ and~$\boldY$ are indistinguishable under certain permitted adversarial noise). 
    Therefore, it follow that $|C_R|=|C|$. 

    Now, suppose for contradiction that $|C|>\frac{q^M}{\sum_{i=0}^{\delta}\binom{M}{i}(q-1)^i}$, i.e., that~$C$ violates the sphere-packing bound for codes of length~$M$ and minimum distance~$\mu$.
    Due to this violation, there exists two codewords~$\boldX_R,\boldY_R\in C_R$ of Hamming distance at most~$\mu-1$.
    Hence, it follows that an adversary can flip at most~$\delta$ bits~$\boldX_R$ and at most~$\delta$ bits in~$\boldY_R$ so as to make them identical, again contradicting the fact that~$C$ is a $(q,n,M,h,\delta)$-forensic code.
    % Let the fragment be a fixed shape at the top left corner of an $n\times n$ codeword, then exists at most $q^M$ codewords due to the pigeonhole principle when there is no bit-flips.
    % As there exists at most $\delta$ bit-flips in the codewords, we may assume all bit-flips happened on the fragment and hence required the Hamming distance between the fragment at least $d$.
    % Thus, we have $|C|\le \frac{q^M}{\sum_{i=0}^{\delta}\binom{M}{i}(q-1)^i}$ is the maximum number of codewords which can correct at most $\delta$ bit-flips due to the sphere packing bound.
    Hence, it follows that the rate of an $(q,n,M,h,\delta)$-forensic code is not greater than 
    \begin{align*}
        \frac{\log|C|}{M}&\le \frac{M-\log(\sum_{i=0}^{\delta}\binom{M}{i}(q-1)^i)}{M}. \tag*{\qedhere}
    \end{align*}    
\end{proof}
Theorem~\ref{thm: necessary bounds for error correcting} gives necessary bounds for any $(q,n,M,h,\delta)$-forensic code.
Our error-correcting encoding in Section~\ref{section:ECencoding} follows the same VDC-based method as Section~\ref{section:encoding for 2D}.
Since the error-correcting step only introduces redundancy, its achievable rate cannot exceed that of the error-free construction.
Table~\ref{tab:explicit-M-h-delta-rate} lists representative choices of $(M,h,\delta)$ that satisfy the requirements of Section~\ref{section:ECencoding} with relatively small $\delta$, together with the corresponding rates.

\begin{table}[h]
\centering
\caption{Example values of $M$, $h$, $\delta$ and rate (converging toward $1/32$).}
\label{tab:explicit-M-h-delta-rate}
\begin{tabular}{r r r r}
        \toprule
        $M$ & $h$ & $\delta$ & \textbf{rate} \\
        \midrule
        157\,260        & 142 & 1 & 0.012559 \\
        1\,299\,196     & 294 & 1 & 0.021219 \\
        7\,166\,400   & 262 & 1 & 0.031054 \\
        57\,923\,712  & 222 & 2 & 0.031097 \\
        1\,203\,775\,488 & 340 & 2 & 0.031150 \\
        4\,934\,680\,112 & 296 & 1 & 0.031238 \\
        \bottomrule
    \end{tabular}
\end{table}
To prove a corresponding existence bound, we make use of the Lovasz Local Lemma~\ref{lemma:lovasz local lemma}.

\begin{lemma}[Lovasz Local lemma]\label{lemma:lovasz local lemma}
    Let  $A_0, A_1, A_2, \ldots, A_k$ be a sequence of events such that each event occurs with probability at most $p$ and each event is independent of all others except for at most~$d$ of them.
    If $ep(d+1)\leq 1$, where~$e$ is the natural logarithm, then the probability that none of the $A_i$'s occur is positive. 
    %Here, $A_0$ is independent of $\{A_1, A_2,\ldots, A_l\}$ means $P(A_0B_1...B_l)=P(A_0)$ for $B_i$ being $A_i$ or $\bar{A_i}$
\end{lemma}
\begin{theorem}\label{thm: random encoding}
  There exists a $(q, n, M, h,\delta)$-forensic code with rate at least $$\frac{M-2\log_q(M^{1.5}n)-\log_q(\sum_{i=0}^{\delta}\binom{M}{i}(q-1)^i)}{M}.$$% using random encoding.
\end{theorem}
\begin{proof}

     Let~$C$ be a uniformly random code over $\Sigma_q$ with codewords $\boldX_1, \boldX_2, ..., \boldX_{|C|}$. 
     We compute the probability of~$C$ being a $(q, n, M, h,\delta)$-forensic code, and then derive a rate bound by which this probability is positive.
     %We prove the probability of decoding successfully is greater than $0$ under rate at least $$\frac{M-2\log(M^{1.5}n)-\log(\sum_{i=0}^{\delta}\binom{M}{i}(q-1)^i)}{n^2}$$
     
    Recall that a legal fragment contains an axis-parallel rectangle of area $M=a\times b$ and the side lengths~$a,b$ satisfy $\min\{a,b\}\ge h$.
    Also, observe that there are at most~$2\sqrt{M}$ pairs~$(a,b)$ such that~$ab=M$, and let $S\triangleq \{S_1, S_2, \ldots, S_l\}$ with $l\leq 2\sqrt{M}$ be the set containing all such $(a,b)$ pairs, 
    %where $S_i$ is the $i$'th $(a,b)$ pair that 
    ordered lexicographically. 

    To utilize Lemma~\ref{lemma:lovasz local lemma}, we define a collection of events and show that the probability of none of them happening is positive, implying the existence of a certain forensic code~$C$.
    A code is not a valid $(q, n, M, h, \delta)$-forensic code if there exist two codewords $\boldX_a \ne \boldX_b$ and two axis-aligned rectangular fragments of area $M$ of the same shape, such that the content of the two fragments differs in at most $2\delta$ entries, i.e. the Hamming distance between them is at most $2\delta$.    

For every two distinct codewords $\boldX_a, \boldX_b$, every shape $S_u \in S$, and every $(\zeta, \eta),(\mu, \nu)\in \llbracket n-1 \rrbracket^2$
%valid top-left corner positions $(\zeta, \eta)\in \llbracket n-1 \rrbracket^2$ and $(\mu, \nu)\in \llbracket n-1 \rrbracket^2$ within $\boldX_a$ and $\boldX_b$, respectively, 
let
\[
A(S_u, \boldX_a^{\zeta,\eta},\boldX_b^{\mu,\nu})
\]
be the ``bad'' event that the fragment of shape~$S_u$ from $\boldX_a$ whose top-left corner is~$(\zeta,\eta)$ is of Hamming distance at most~$2\delta$ from the fragment of shape~$S_u$ from $\boldX_b$ whose top left corner is $(\mu,\nu)$.
%and the Hamming distance between the fragments is at most $2\delta$.
%Since the total area is $M$ and with 
By a known formula for the size of the Hamming ball, we have that
%the probability of this event is upper bounded by:
\begin{align}\label{equation:Pr}
\Pr(A(S_u, \boldX_a^{\zeta,\eta},\boldX_b^{\mu,\nu})) = \sum_{i=0}^\delta \binom{M}{i} \cdot \frac{(q-1)^i}{q^M}.
\end{align}
for all $a\ne b,u,\zeta,\eta,\mu,\nu$ in their corresponding ranges. 
%Since this probability is identical for all parameters, we omit all indices and write~$A=A(S_u, \boldX_a^{\zeta,\eta},\boldX_b^{\mu,\nu})$.
%Hence, we define $\Pr(A)= \sum_{i=0}^\delta \binom{M}{i} \cdot \frac{(q-1)^i}{q^M}$ be the exact bound of all possible events.     

    For a given $A(S_u,\boldX_a^{\zeta,\eta},\boldX_b^{\mu,\nu})$ let $\{A(S_{u_i},\boldX_{a_i}^{\zeta_i,\eta_i},\boldX_{b_i}^{\mu_i,\nu_i})\}_{i=1}^{\kappa}$  
    %with $u_i\in \{1, 2, \dots, l\}$, $\zeta_i,\eta_i,\mu_i,\nu_i\in \llbracket n-1 \rrbracket$ and $a_i\ne b_i, a_i,b_i\in\{1,2,\dots ,|C|\}$ for all~$i$, 
    be all events which are dependent of it.
    These are precisely all events which involve fragments that intersect nontrivially with the fragments in $A(S_u,\boldX_a^{\zeta,\eta},\boldX_b^{\mu,\nu})$.
    %The dependent events of $E$ are the events that there are intersections for the fragments with $E$, i.e. $S_{u_i}$-shape located at $\boldX_{a_i}^{\zeta_i,\eta_i}$ or $\boldX_{b_i}^{\mu_i,\nu_i}$ intersect with $S_{u}$ located at $\boldX_{a}^{\zeta,\eta}$ or  $\boldX_{b}^{\mu,\nu}$ . 
    With $S_{u}$ located at $\boldX_{a}^{\zeta,\eta}$, there are at most $9M^2$ entries from $\boldX_a$ where another $S_{u_i}$ can be placed so that the intersection is nonempty.
    %such that $S_{u}$ located at $\boldX_{a}^{\zeta,\eta}$ intersects $S_{u_i}$ located at $\boldX_{a}^{\zeta_i,\eta_i}$. 
    Specifically, these $9M^2$ entries from $\boldX_a$ are $[\zeta-M, \zeta+2M]\times[\eta-M,\eta+2M]$.
    
    As there are~$l$ different rectangle shapes, $S_{u}$ located at $\boldX_{a}^{\zeta,\eta}$ intersects at most $l\cdot9M^2$ $S_{u_i}$'s located at $\boldX_{a}^{\zeta_i,\eta_i}$.
    Additionally, the number of different $\boldX_{b_i}^{\mu_i,\nu_i}$'s is at most $n^2\cdot (|C|-1)$.
    %from $|C|-1$ different $b_i$'s, $n$ different $\mu_i$ and $\nu_i$.
    Therefore, it follows that
    \begin{align}\label{equation:kappaBound}
        \kappa\le2\cdot 9M^2\cdot (|C|-1)\cdot n^2\cdot l,
    \end{align}
    where the~$2$ factor is due to symmetry.
    %As the same number applied to dependent events for $S_{u}$ located at $\boldX_{b}^{\mu,\nu}$ due to symmetry, we have $\kappa<2\cdot 9M^2\cdot (|C|-1)\cdot n^2 \cdot l$.
    
%\noindent\textbf{Paragraph: Applying Lovasz Local Lemma}

    By Lemma~\ref{lemma:lovasz local lemma} and by~\eqref{equation:Pr}, if
    \begin{align}\label{equation:LLLcondition}
        e\cdot\left( \sum_{i=0}^\delta \binom{M}{i} \cdot \frac{(q-1)^i}{q^M}\right) \cdot(\kappa+1)<1,
    \end{align}
    then there exists a positive probability for none of the bad events to occur, i.e., there exists a forensic code for these parameters. 
    Indeed, according to~\eqref{equation:kappaBound} we have that~\eqref{equation:LLLcondition} is satisfied for a random code~$|C|$ no larger than
    \begin{align*}
        \frac{q^M}{49(Mn)^2l\sum_{i=0}^{\delta}\binom{M}{i}(q-1)^i}.
    \end{align*}
    % In order to show the existence of a certain forensic code $C$, the parameters need to satisfy the condition $e\cdot\Pr(A)\cdot(\kappa+1)<1$ from Lovasz Local Lemma~\ref{lemma:lovasz local lemma}
    % This gives us $$|C|<\frac{q^M}{49(Mn)^2l\sum_{i=0}^{\delta}\binom{M}{i}(q-1)^i}.$$
    %Thus, there exists a $(q,n,M,h,\delta)$-forensic code with size $\frac{q^M}{49(Mn)^2l\sum_{i=0}^{\delta}\binom{M}{i}(q-1)^i}-1$ as the probability of none of the bad events happens is positive.
    The rate of the resulting code is therefore
    \begin{align*}
        \frac{\log_q|C|}{M}=\frac{M-2\log_q(M^{1.25}n)-\log_q(\sum_{i=0}^{\delta}\binom{M}{i}(q-1)^i)}{M}. \tag*{\qedhere}
    \end{align*}
       
\end{proof}
%Theorem~\ref{thm: random encoding} shows that there exists a~$(q,n,M,h,\delta)$-forensic code which reach the upper bound when~$\frac{\log_q(n)}{M}\overset{n\to\infty}{\longrightarrow}0$.
Theorem~\ref{thm: random encoding} readily implies that in the parameter regime where~$\frac{\log_q(n)}{M}\overset{n\to \infty}{\to} 0$ (which also implies that~$M\to\infty$), there exists a code which attains the rate upper bound given in Theorem~\ref{thm: necessary bounds for error correcting}.

Furthermore, it is known~\cite[Thm.~4.9]{roth2006introduction} that the asymptotic rate which follows from the sphere-packing bound is at most~$1-H_q(\delta/M)+o(1)$, where~$H_q$ is the base-$q$ entropy function and~$o(1)$ is a function which goes to zero with~$M$.
Recall that our error-free encoding scheme achieves rate of~$1/32$ (Theorem~\ref{theorem:comparison of both methods}), and our error-resilient one in Section~\ref{section:ECencoding} follows similar ideas, hence its rate cannot exceed~$1/32$ as well.
%the same idea as in Section~\ref{section:encoding for 2D} for error-free encoding.
%Hence, the code rate in Section~\ref{section:ECencoding} cannot exceed the one in Section~\ref{section:encoding for 2D}, which is $\frac{1}{32}$ (Theorem~\ref{theorem:comparison of both methods}).
Consequently, finding codes of rate~$1$, which exist according to the above discussion, is left for future work.

\ifFULL
\section{Open problems and Future work}
In this paper, we developed codes based on the van der Corput set and the Halton–Hammersley set in both two and three dimensions, as well as in noisy settings.  
These constructions all achieve a constant rate.  
In contrast, the existence result in Theorem~\ref{thm: random encoding}, which leverages the Lovasz Local Lemma, shows that codes of rate~$1$ exist in certain parameter regimes.
Thus, a significant gap remains between the performance of our explicit constructions and the existential bound which may be improved in future work.

\bibliography{sn-bibliography}

%\begin{appendices}
\appendix

\section*{Appendix}
\section{Omitted proofs}\label{section:omittedProofs}

\begin{lemma}\label{app:proof-lemma:valid coloring as vdc set}
    Every color set induced by the coloring function~\eqref{coloring function} is of the form~$T_i$ mentioned in~\eqref{equation:T_i}, i.e. $L_i=T_i$ for all $i\in \llbracket m'-1\rrbracket$.
\end{lemma}
\begin{proof}
     %We first recall several definitions from previous sections. 
     Recall that:
     \begin{itemize}
         \item The $m'$-VDC set is~$C_{m'}\triangleq\{(q,f_{m'}(q))\}_{q=0}^{m'-1}$, see Definition~\ref{definition:VDCset}, and for~$x,y\in\llbracket n/(dm')-1\rrbracket$ we have the~$(xm',ym')$-shifted VDC-set
         $$C_{m'}^{(x,y)}\triangleq \{(z_1+xm',z_2+ym')\vert (z_1,z_2)\in C_{m'}\},$$ see Lemma~\ref{union of vdc lemma}.
         \item The union of all shifted VDC-sets is $T\triangleq \cup_{x,y\in \llbracket n/(dm')-1\rrbracket}C_{m'}^{(x,y)}$, and its shifted version is
         $$T_k\triangleq\{(u,(v+k)\bmod n)\vert (u,v)\in T \}$$ for any~$k\in\llbracket m' -1\rrbracket$. 
     \end{itemize}
     Combining these definitions gives us that for any~$k\in\llbracket m'-1\rrbracket$,
     \begin{align}\label{equation: updated Tk}
         T_k=\{(q+xm',(f_{m'}(q)+ym'+k)\bmod \textstyle\frac{n}{d})~\vert
         ~x,y\in\llbracket n/(dm')-1\rrbracket\mbox{ and }
          q\in\llbracket m'-1\rrbracket\}.
     \end{align} 
     %For any given color $k\in \llbracket m'\rrbracket$, p
     Next we show that all entries of~$T_k$ are colored in color~$k$, for all~$k\in\llbracket m-1'\rrbracket$.
     Indeed, applying~$c$ over an element of~$T_k$~\eqref{equation: updated Tk} yields
     \begin{align}\label{equation:cijasTk}
         c(q+&xm',(f_{m'}(q)+ym'+k)\bmod n)=\nonumber\\
         &((f_{m'}(q)+ym'+k)\bmod n-f_{m'}(q))\bmod m'
     \end{align}
    %Then we perform a case analysis for the above equation. 
    Therefore, to prove~$T_k\subseteq L_k$ it suffices to show that~$\eqref{equation:cijasTk}=k$, for which  we split to cases.
    \begin{itemize}
        \item Case~1: $y\in \llbracket n/(dm')-2\rrbracket$.
            Since~$f_{m'}(q)+ym'+k\le \frac{n}{d}$ for~$y\in \llbracket n/(dm')-2\rrbracket$, it follows that
            \begin{align*}
            ((f_{m'}(q)+ym'+k)\bmod \textstyle\frac{n}{d}-f_{m'}(q))\bmod m' =(f_{m'}(q)+ym'+k-f_{m'}(q)) \bmod m'= k.
            \end{align*}
        \item Case 2: $y=\frac{n}{dm'}-1$. 
            \begin{itemize}
                \item Subcase 1: If $f_{m'}(q)+ym'+k\ge \textstyle\frac{n}{d}$ then
                \begin{align*}
                    ((f_{m'}(q)+&ym'+k)\bmod \textstyle\frac{n}{d}-f_{m'}(q))\bmod m' =(f_{m'}(q)-m'+k-f_{m'}(q)) \bmod m'= k.
                \end{align*}
                \item Subcase 2: If $f_{m'}(q)+ym'+k < \textstyle\frac{n}{d}$ then
                \begin{align*}
                    &((f_{m'}(q)+ym'+k)\bmod \textstyle\frac{n}{d}-f_{m'}(q))\bmod m'=\\
                    &(f_{m'}(q)+(\textstyle\frac{n}{dm'}-1)m'+k-f_{m'}(q)) \bmod m'= k.
                \end{align*}
            \end{itemize}
    \end{itemize}
     Therefore, it follows that~$T_k\subseteq L_k$. 
     To show that~$L_k\subseteq T_k$, let~$(i,j)\in L_k$, i.e.,~$c(i,j)=k$. By~\eqref{coloring function}, we have $(j-f_{m'}(i\bmod m'))=k+ym'$  for some~$y\in \mathbb{Z} $. Therefore, 
     \begin{align*}
         j=f_{m'}(i\bmod m'))+k+ym'=(f_{m'}(i\bmod m'))+k+ym')\bmod \textstyle\frac{n}{d},
     \end{align*}
     where the last equality follows since~$j\in\llbracket n/d -1\rrbracket$.
     %then $j=(f_{m'}(i\bmod m')+ym'+k)\bmod \textstyle\frac{n}{d}$. 
     For~$q\triangleq i \bmod m'$ we have~$i=q+xm'$ for some~$x\in\mathbb{Z}$. As $i,j\in \llbracket n/d-1\rrbracket $, it follows that range of $x,y$ can be reduced to $\llbracket n/(dm')-1\rrbracket $, i.e. $x,y \in \llbracket n/(dm')-1\rrbracket$.
     Therefore~$(i,j)=(q+xm',(f_{m'}(q)+ym'+k)\bmod \textstyle\frac{n}{d})$ for~$x,y\in \llbracket n/(dm')-1 \rrbracket$ and~$q\in\llbracket m'-1\rrbracket$,
     and therefore~$(i,j)\in T_k$ by~\eqref{equation: updated Tk}.
\end{proof}

\begin{lemma}\label{lemma:parametrsAreGoodappendix}
The parameters in~\eqref{equation:parameters} satisfy the requirements in Theorem~\ref{Theorem: cited sliced channel}.    
\end{lemma}

\begin{proof}
    Let $L'\triangleq 3\log(Q)+4\delta^2+2$. 
    We need to prove $$L'+4\delta L'+2\delta\log(4\delta L')\leq L.$$ 
    Recall from Section~\ref{section:overview} that $d$ is the largest integer such that $h\geq 3d-1$. 
    Thus we have $d\ge \lfloor\frac{h}{3}\rfloor-1$, which implies that 
     $\delta\le c(3d+1)^{2/3}\le3cd^{2/3} $ when $d$ is sufficiently large. 
     Recall that $m$ is the largest integer such that \( M \geq (4d-1)((m+1)d+d-1) \) which implies $M\geq md^2$.
     Thus we have $md^2\leq M=h^22^{ch^{4/3}}\le (3d+1)^22^{c(3d+1)^{4/3}}$. 
     When $d$ is sufficiently large, we then have $m\le 9\cdot2^{3cd^{4/3}}$.
     Thus, $Q\le 2^{3cd^{4/3}}$ and hence $\log(Q)\le {3cd^{4/3}}$.
     As $L'\triangleq 3\log(Q)+4\delta^2+2$, we have $L'\le 9cd^{4/3} +4(3cd^{2/3})^2+2=(9c+36c^2)d^{4/3}+2$. 
     When $d$ is sufficiently large and $c$ is sufficiently small, we can rewrite it as $L'\le 10cd^{4/3}$.
     Hence, we have $4\delta L'\le 12cd^{2/3}\cdot 10cd^{4/3}=120c^2d^2$ and $2\delta\log(4\delta L')\le 12cd^{2/3}\cdot\log(120c^2d^2)$
     which implies $L'+4\delta L'+2\delta\log(4\delta L')\le L$ for sufficiently small~$c$. 
\end{proof}
\begin{lemma}\label{Lemma: forInequlaity}
    For~$d$ and~$p$ as defined in the proof of Lemma~\ref{lemma:enough complete units}, we have that
  $((2^{i}+1) d  + d - 1)((2^{p+1-i}+1) d + d - 1) \le (4d-1)((2^p+1)d+d-1)$ for all $i\in \{1, 2, 3, \ldots, p\}$.
\end{lemma}
\begin{proof}Indeed,
 \begin{align*}
     ((2^{i}+1) d  + d - 1)((2^{p+1-i}+1) d + d - 1) &= (2^{i} d  + 2d - 1)(2^{p+1-i} d + 2d - 1)\\
     &=2^{p+1}d^2+d(2d-1)(2^{i}+2^{p+1-i})+(2d-1)^2\\
     &\le 2^{p+1}d^2+d(2d-1)(2+2^{p})+(2d-1)^2\\
     &=(4d-1)((2^p+1)d+d-1).&\qedhere
 \end{align*}
\end{proof}

%\end{appendices}


\end{document}